\documentclass[a4paper,onecolumn,10pt]{article}
\pdfoutput=1
 
\usepackage[utf8]{inputenc}
\usepackage[english]{babel}
\usepackage[T1]{fontenc}
\usepackage{amsmath}
\usepackage{amssymb}
\usepackage{amsthm}
\usepackage[margin=1in]{geometry}
\usepackage{authblk}
\usepackage{orcidlink}
\usepackage{hyperref}
\usepackage{mathrsfs}
\usepackage{tikz-cd}
\usepackage{listings}
\usepackage{mhchem}
\usepackage[T1]{fontenc}

\usepackage{siunitx}
\usepackage{upgreek}
\usepackage{amsmath, amssymb, amsthm, mathtools, amssymb}
\usepackage{geometry}
\usepackage{enumitem}
\usepackage{tikz-cd}
\usepackage[dvipsnames]{xcolor}
\usepackage{booktabs}
\usepackage{array}
\usepackage{tcolorbox}
\tcbuselibrary{listings}
\usepackage{microtype}
\usepackage{titlesec}
\usepackage{parskip}
\usepackage{fancyhdr}
\usepackage{caption}
\usepackage{bm}
\usepackage{chemformula}
\usepackage{makecell}

\DeclareMathOperator{\FP}{\textit{F}_P}
\newcommand{\LGraph}{\mathbf{LGraph}}
\newcommand{\LGraphP}{\mathbf{LGraph}_P}

\newcommand{\Gstar}{G^*}

\newcommand{\Ce}{{C_e}}
\newcommand{\Hel}{{\mathcal{H}_{\mathrm{el}}}}
\newcommand{\FV}{F_{\!V}}
\newcommand{\OrbMorse}{\mathbf{Orb}^{\mathrm{Morse}}}

\newcommand{\Xseam}{\mathcal{X}_{01}}

\newcommand{\SE}{\mathrm{SE}}       
\newcommand{\Autmu}{\mathrm{Aut}_{\mu}} 
\newcommand{\HilbBundR}{\mathbf{HilbBund}_{\mathbb{R}}} 
\newcommand{\Lzero}{L_0}         

\hypersetup{
  colorlinks=true, linkcolor=blue!60!black,
  citecolor=green!50!black, urlcolor=blue!80!black
}

\theoremstyle{plain}
\newtheorem{theorem}{Theorem}[section]
\newtheorem{proposition}[theorem]{Proposition}

\newtheorem{corollary}[theorem]{Corollary}

\theoremstyle{definition}
\newtheorem{definition}[theorem]{Definition}

\theoremstyle{remark}
\newtheorem{remark}[theorem]{Remark}

\usepackage{longtable}   

\tcbuselibrary{theorems, breakable, skins}

\newtcolorbox{forcingbox}[1][Forcing: why the next level is necessary]{
  colback=orange!7!white, colframe=orange!65!black,
  boxrule=0.9pt, arc=4pt, left=8pt, right=8pt, top=6pt, bottom=6pt,
  title={#1}, fonttitle=\bfseries\small, breakable
}

\newtcolorbox{chembox}[1][For the chemist]{
  colback=green!5!white, colframe=green!50!black,
  boxrule=0.9pt, arc=4pt, left=8pt, right=8pt, top=6pt, bottom=6pt,
  title={#1}, fonttitle=\bfseries\small, breakable
}

\newtcolorbox{mathbox}[1][For the mathematician]{
  colback=blue!5!white, colframe=blue!50!black,
  boxrule=0.9pt, arc=4pt, left=8pt, right=8pt, top=6pt, bottom=6pt,
  title={#1}, fonttitle=\bfseries\small, breakable
}

\newtcolorbox{insightbox}[1][Key insight]{
  colback=violet!5!white, colframe=violet!50!black,
  boxrule=0.9pt, arc=4pt, left=8pt, right=8pt, top=6pt, bottom=6pt,
  title={#1}, fonttitle=\bfseries\small, breakable
}

\newtcolorbox[auto counter, number within=section]{openprob}[1][]{
       colback=red!4!white, colframe=red!60!black,
       fonttitle=\bfseries,
       title={Open Problem~\thetcbcounter\ifstrempty{#1}{}{:~#1}},
       #1}

\newcommand{\Lk}{\mathcal{L}}
\newcommand{\NN}{\mathbb{N}}
\newcommand{\ZZ}{\mathbb{Z}}
\newcommand{\RR}{\mathbb{R}}
\newcommand{\Sp}{\mathcal{S}}          
\newcommand{\Rx}{\mathcal{R}}          
\newcommand{\Mon}{\mathbb{N}[\Sp]}     
\newcommand{\id}{\mathrm{id}}

\newcommand{\Aut}{\mathrm{Aut}}
\newcommand{\coker}{\mathrm{coker}}
\newcommand{\im}{\mathrm{im}}
\newcommand{\rank}{\mathrm{rank}}
\newcommand{\FH}{F_{\!H}}             
\newcommand{\FG}{F_{\!G}}             
\newcommand{\FS}{F_{\!S}}             

\newcommand{\Ob}{\mathrm{Ob}}

\lstdefinelanguage{Haskell2}{
  language        = Haskell,
  sensitive       = true,
  morecomment     = [l]{--},
  morecomment     = [s]{\{-}{-\}},
  morestring      = [b]",
  morekeywords    = {
    data, newtype, type, class, instance, where, let, in, case, of,
    do, forall, module, import, qualified, as, hiding, deriving,
    stock, generic, IO, Map, Set, Int, Double, Bool, String,
    Maybe, Either, Eq, Ord, Show, Generic,
    ParaMorphism, SimConfig, Population, PhysicsParams,
    Reaction3, SimState, Species, Nucleotide, DepMap, Network3,
    Bag, Complex, Reaction0, GillespieEvent,
    true, false, Nothing, Just, Left, Right
  },
  keywordstyle    = \color{blue!80!black}\bfseries,
  commentstyle    = \color{gray!60}\itshape,
  stringstyle     = \color{teal},
  basicstyle      = \ttfamily\small,
  breaklines      = true,
  keepspaces      = true,
  columns         = flexible,
  showstringspaces= false,
  literate        =
    {->}{{$\to$}}2
    {=>}{{$\Rightarrow$}}2
    {<-}{{$\leftarrow$}}2
    {>=}{{$\geq$}}2
    {<=}{{$\leq$}}2
    {/=}{{$\neq$}}2
    {>>}{{$\gg$}}2
    {>>=}{{$\gg\!\!=$}}3
    {\\}{{$\lambda$}}1
    {forall}{{$\forall$}}1
    {alpha}{{$\alpha$}}1
    {beta}{{$\beta$}}1
    {->}{{$\to$}}2
}

\definecolor{codegray}{HTML}{263238}   
\definecolor{codebg}{HTML}{FAFAFA}     

\usepackage{tcolorbox}
\tcbuselibrary{listings,breakable,skins}

\newtcblisting{haskellbox}[1][]{%
  enhanced,
  breakable,
  listing only,
  listing engine=listings,
  listing options={
    language=Haskell2,
    basicstyle=\ttfamily\small,
    xleftmargin=2pt,
    aboveskip=0pt,
    belowskip=0pt,
  },
  colback=codebg,
  colframe=codegray!60,
  coltitle=white,
  colbacktitle=codegray!85,
  fonttitle=\bfseries\small\sffamily,
  title={Haskell},
  attach boxed title to top left={yshift=-2mm,xshift=4mm},
  boxed title style={sharp corners,colback=codegray!85},
  sharp corners=south,
  left=4pt,right=4pt,top=8pt,bottom=4pt,
  #1
}

\newtcblisting[]{haskellsnip}[1][]{%
  breakable,
  listing only,
  listing engine  = listings,
  listing options = {
    language    = Haskell2,
    basicstyle  = \ttfamily\footnotesize,
    xleftmargin = 2pt,
    aboveskip   = 0pt, belowskip = 0pt,
  },
  colback  = codebg,
  colframe = codegray!40,
  sharp corners,
  left = 3pt, right = 3pt, top = 3pt, bottom = 3pt,
  #1
}

\providecommand{\Petri}{\mathbf{Petri}}
\providecommand{\PermCat}{\mathbf{PermCat}}
\providecommand{\BR}{\mathrm{B}\RR}            
\providecommand{\gen}{\mathfrak{g}_S}          
\providecommand{\Bg}{\mathrm{B}\mathfrak{g}_S} 
\providecommand{\AutP}{\Aut_{\Petri}}
\providecommand{\Symm}{\mathrm{Sym}}
 
\title{Category Theoretic Framework for Chemistry I: \\ a Tower of Chemistry}
 
\author[1]{Kyunghoon Han\thanks{\texttt{kyunghoon.h@gmail.com}}\,\orcidlink{0000-0001-7670-300X}}
\affil[1]{Department of Physics and Materials Science, University of Luxembourg, Luxembourg}
 
\date{}
 
\begin{document}
 
\maketitle
 
\footnotetext[1]{The author thanks Dr. David Luong of Carleton University for his thorough reading of the theory and his valuable comments, Mr. Giuseppe Mansi of the University of Luxembourg for his helpful comments on Section 2.4, and Dr. Gregory Fonseca of the University of Luxembourg for his valuable comments on the abstract.}

\begin{abstract}
Many laws of chemistry are exact within a limited scope and acquire a separate caveat outside it,
and the caveats are usually treated as unrelated.
This work argues that they share one cause.
Each caveat marks a point where a question is asked of a description too coarse to answer it:
the question belongs to a richer level of structure than the description carries.
To make these levels explicit,
the paper builds a tower of categories over the free symmetric monoidal category of a Petri net,
the simplest categorical presentation of a reaction network.
The levels, from the bottom up, are stoichiometry, thermochemistry, equilibrium, reaction kinetics,
reaction mechanism, molecular geometry, and electronic structure.
Each adds one kind of chemical content over the level below,
and a forgetful functor runs back down.
One question runs through the tower:
\emph{what can a level express that the level below cannot}?
Answering it places each measurable quantity at the level where it lives,
and identifies the content the levels below could record but not account for.
The method is then turned on with two pieces of known chemistry.
It recasts a classical criterion for when a reaction network has a unique stable equilibrium,
the deficiency-zero theorem,
as the rigidity of a single forgetful fibre.
It also follows one familiar reaction up the tower:
the ring opening that the Woodward–Hoffmann rules govern.
At each level,
from stoichiometry to electronic structure,
the reaction becomes a distinct categorical object.
\end{abstract}

\section{Introduction}
\label{sec:intro}

Chemistry is a science of rules with exceptions.
Hess's law~\cite{Hess1840} adds reaction enthalpies along any path,
except that the sum depends on a choice of reference unless the enthalpy is a state function in the sense of Gibbs~\cite{Gibbs1875};
detailed balance~\cite{Lewis1925} constrains equilibrium rate constants,
except that the constraints split into independent families on networks with cycles,
as Wegscheider observed~\cite{Wegscheider1901} and as the modern theory of detailed-balanced and complex-balanced systems makes precise~\cite{CraciunJinYu2020};
the Woodward--Hoffmann rules~\cite{WoodwardHoffmann1965, WoodwardHoffmann1969} predict the stereochemistry of a pericyclic reaction,
except that they invert between thermal and photochemical regimes,
an inversion now read as topological~\cite{DaggettYangLiuMuechler2024}.
Each rule is taught with a list of caveats to be learned separately.
The thesis organising this paper is that the caveats are not accidents of chemistry but artefacts of description:
each arises when a question that belongs to a richer level of mathematical structure is posed in the vocabulary of a poorer one.
A heat asked to be path-independent is a question about coboundaries posed about a mere additive label.
An orbital-symmetry selection rule asked to be regime-independent is a question about the topology of an electronic bundle posed about a scalar energy,
a geometric phase whose chemical effect Longuet-Higgins predicted~\cite{LonguetHiggins1975} and recent scattering experiments have observed as an interference shift~\cite{Yuan2020}.
Make the levels explicit,
with category theory as the natural language in which to do so~\cite{Akitsu2023},
and each exception becomes the visible signature of a structure the working vocabulary could not express.

This paper makes the levels explicit for reaction networks,
as a graded tower of categories.
The point of entry is one that will be familiar to readers of this journal:
a reaction network is already a categorical object.
Following Meseguer and Montanari~\cite{MeseguerMontanari},
a Petri net presents a free symmetric monoidal category whose objects are multisets of species and whose morphisms are the reactions and their formal composites and tensors;
the chemistry of a network is then a sequence of functors out of this free category.
Baez and Pollard~\cite{BaezPollard} equip the kinetic layer of this picture with the master equation and make open networks compose by pushout,
in the tradition of decorated and structured cospans~\cite{Fong2016Thesis, baez2019structured}
and of open Petri nets~\cite{BaezMaster2020}.
The present work takes the same syntactic base and asks a different question:
not how networks are composed horizontally,
but how the \emph{depth} of chemical description stacks vertically.
Above the bare network, sit thermochemistry, equilibrium, kinetics, mechanism, molecular geometry,
and electronic structure,
and the claim is that these are not six incomparable formalisms but six levels of one object,
each obtained from the one below by adjoining exactly one kind of data and each equipped with a forgetful functor back down.

\paragraph{The construction in one paragraph.}
Over a Petri net $P=(\Sp,\Rx,s,t)$ the tower is a diagram of categories $\Lk_0(P),\Lk_1(P),\dots,\Lk_6(P)$ with forgetful functors descending it.
The base $\Lk_0(P)$ is the free symmetric monoidal category on the net.
The lower levels decorate it with numerical functors, reaction enthalpy at $\Lk_1$,
free energy and a reversal involution at $\Lk_2$, mass-action rate data at $\Lk_3$,
each a functor out of $\Lk_0(P)$ into a one-object target,
so that the level is a pair (network, decoration) and the forgetful functor drops the decoration.
The mechanistic level $\Lk_4(P)$ is not a further decoration but a change of underlying category:
it replaces the transitions of the net by double-pushout graph rewriting on molecular structures,
in the adhesive-category sense~\cite{LackSobocinski},
each rewrite labelled by the reaction it realises.
A short stereochemical half-step refines it,
and above it the geometric level $\Lk_5(P)$ assigns each molecule a potential energy surface on a configuration orbifold,
while the electronic level $\Lk_6(P)$ records the topology of the electronic states over that surface.
A seventh,
fully quantum level would close the tower in principle but is left open here,
for reasons given below.

\paragraph{One question, several answers.}
What makes the tower a single structure rather than a list is that one question is asked at every step:
what can a level express that the level below cannot?
Each level is a category whose structure-preserving self-equivalences form an automorphism group,
and each forgetful functor $U_k\colon\Lk_k(P)\to\Lk_{k-1}(P)$ induces by restriction a homomorphism $\varphi_k\colon\Aut(\Lk_k(P))\to\Aut(\Lk_{k-1}(P))$.
The new content of level $k$ is what this homomorphism fails to capture,
and across the tower that failure takes three distinct forms.
For the decorated levels $\Lk_1$ to $\Lk_3$ it is a \emph{broken symmetry}:
$\varphi_k$ is not surjective, its cokernel $\coker(\varphi_k)$ is a pointed set,
and a non-trivial element is a base relabelling that the new decoration breaks,
a pair of reactions the level below is forced to identify and the level above separates.
For the mechanistic level $\Lk_4$ the underlying category changes rather than acquiring a decoration,
and the new content is certified not by a surviving symmetry but by an \emph{inequality of derivations}:
one effective reaction is realised by graph-rewriting derivations of different length,
a distinction no lower observable records.
For the geometric and electronic levels $\Lk_5$ and $\Lk_6$ it is an \emph{object refinement}:
a single lower object splits into distinct upper ones,
two isotopologues over one bond graph, two electronic lifts over one potential surface,
and the new datum is what the lower level could record as a free number but could not derive.
The three answers are different,
but the question is one,
and what it locates changes as one climbs:
a coboundary obstruction at energetics, a reversal involution at equilibrium,
a derivation length at reaction mechanism, a mass-weighted metric at geometry,
a $\ZZ/2$ Berry class at electronic structure.
This is the sense in which the levels are forced and not merely posited:
each is the minimal extension that breaks a symmetry, splits a morphism,
or refines an object the level below could not tell apart.

\paragraph{What this paper proves, and what it imports.}
This is a self-contained categorical paper,
extracted from and consistent with a longer manuscript~\cite{Han2026} but standing on its own.
It carries out the construction and the forcing diagnostic \emph{in full} for the base through the mechanistic level,
$\Lk_0$ to $\Lk_4$:
the base automorphism group is computed, its semidirect-product structure proved,
the broken-symmetry cokernel for each decorated level identified as a coset,
and the derivation-length separation at the mechanistic level established directly.
For the geometric and electronic levels the target categories are more elaborate,
a category of Morse triples and a category of electronic bundles carrying a $\ZZ/2$ sign class,
and here the paper changes register.
It states their constructions,
identifies the datum each level forces and the object refinement that certifies it,
and imports from~\cite{Han2026} the proofs that those targets are well-defined and that the corresponding cokernels are non-trivial.
The boundary between what is proved here and what is imported is marked explicitly wherever it is crossed.
The seventh level is genuinely open:
this paper does not construct it, and states it only as a horizon.
The paper closes with a placement procedure that turns the tower into a working method.
Given a physical quantity,
the procedure returns the level at which the quantity lives and what the tower then guarantees about it.
It is run twice: once to recast Feinberg's Deficiency Zero Theorem~\cite{HornJackson, Feinberg} as the rigidity of a single forgetful fibre,
exposing two textbook misuses as level-mismatch errors,
and once on the cyclobutene--butadiene electrocyclic reaction,
the standard Woodward--Hoffmann system~\cite{WoodwardHoffmann1965},
carried up all six objects it becomes.

\paragraph{Prerequisites and conventions.}
The reader is assumed comfortable with symmetric monoidal categories, functors,
and free constructions at the level of~\cite{MacLane1998, Leinster2014};
chemical terms are glossed on first appearance,
so no chemistry background is required.
All categories of decorations are strict,
and ``functor'' between monoidal categories means strict symmetric monoidal functor unless stated otherwise.
A \emph{Petri net}
is $P=(\Sp,\Rx,s,t)$ with $\Sp$ a finite set of species, $\Rx$ a finite set of
reactions, and $s,t\colon\Rx\to\Mon$ source and target maps into the free
commutative monoid $\Mon=\NN[\Sp]$ on the species; this is the sole input to the
tower, and every level is functorial in it. Automorphisms are taken in the relevant
category of structured objects, so that $\Aut(\Lk_k(P))$ means the
presentation-preserving self-equivalences of the level, and cokernels of group
homomorphisms with non-normal image are read as pointed sets of cosets. The next
section fixes the base category and develops the lower tower in detail;
the reader already fluent in the Petri-net as \emph{free symmetric monoidal category} dictionary may begin there and refer back to this introduction only for the diagnostic.
\section{Theory: the tower of chemistry}
\label{sec:theory}

A single chemical reaction can be described at many depths.
The same transformation is an entry in a stoichiometric table,
a heat of reaction, a rate constant, a rearrangement of bonds, a motion across a potential energy surface,
or a passage between electronic states, depending on the level of detail one chooses to track.
Chemistry supplies a separate formalism for each.
This section builds these descriptions as the levels of one structure,
a tower of categories, all generated from a single combinatorial input.

Chemistry can be described by analogy to mathematical functions:
the input is one or more chemical species, the function is a chemical reaction,
and the output is again one or more chemical species.
Underpinning all of chemistry, beginning at its most basic level,
is the constraint that the number of atoms of each element in the input must equal the number in the output.
This kind of bookkeeping is called \emph{stoichiometry}.
In the combustion of hydrogen,
\[
\ch{2 H2 + O2 -> 2 H2O},
\]
two molecules of hydrogen and one of oxygen are consumed and two of water is produced,
and every hydrogen and oxygen atom is accounted for on both sides.
The input here is not a set of species, but a \emph{multiset}:
two copies of $\ch{H2}$ are consumed, not one.
The mathematical object that exactly records this datum in the simplest sense
(species, reactions, and the multiset each reaction consumes and produces)
is a \emph{Petri net}, which is defined as follows.

\begin{definition}[Petri net]
\label{def:petri}
A \emph{Petri net} is a tuple $P=(\Sp,\Rx,s,t)$ consisting of a finite set $\Sp$ of \emph{species},
a finite set $\Rx$ of \emph{reactions}, and a pair of \emph{source} and \emph{target} maps
\[
s,t\colon\Rx\longrightarrow\Mon ,
\]
where $\Mon=\NN[\Sp]$ is the free commutative monoid on $\Sp$, written additively.
For a reaction $r\in\Rx$,
the multiset $s(r)$ is its \emph{input} and $t(r)$ its \emph{output}.
A \emph{morphism of Petri nets} $P\to P'$ is a pair of maps $\Sp\to\Sp'$ and $\Rx\to\Rx'$ commuting with the source and target maps;
the resulting category is denoted $\Petri$.
\end{definition}

An element of $\Mon$ is a finite multiset of species, a formal nonnegative-integer
combination of the species such as $2\,\ch{H2}+\ch{O2}$; the free commutative monoid
is precisely what supplies both the coefficient and the unordered sum. Two features
of the definition are worth isolating before construction proceeds, each
illustrated by a standard example.

The first is why the codomain of $s,t$ must be $\Mon$ instead of $\Sp$.
In the hydrogen combustion above,
the input $2\,\ch{H2}+\ch{O2}$ is a multiset with a repeated species and is named by no single element of $\Sp$:
it is the multiset of reactants, not any one of these molecules, that a reaction consumes.
A set-valued source map could not record the coefficient $2$,
so the free commutative monoid is forced by chemistry.

The second is that $\Rx$ is a set of \emph{labelled} reactions,
not merely a relation between multisets,
so that distinct reactions may share an input and an output.
The Michaelis--Menten mechanism of enzyme catalysis,
\[
\ch{E + S <=> ES -> E + P},
\]
illustrates this together with reversibility:
an enzyme $\ch{E}$ and a substrate $\ch{S}$ bind reversibly to a complex $\ch{ES}$,
which then releases the enzyme and a product $\ch{P}$.
Here $\Sp=\{\ch{E},\ch{S},\ch{ES},\ch{P}\}$ and $\Rx$ have three reactions:
binding, with source $\ch{E}+\ch{S}$ and target $\ch{ES}$;
its reverse, with source and target exchanged; and catalysis,
with source $\ch{ES}$ and target $\ch{E}+\ch{P}$.
Forward and reverse binding reactions are distinct elements of $\Rx$ sharing an input--output pair up to direction,
and recording reversibility means keeping them as separate labels rather than as a single undirected edge.

Throughout,
$\Sp$ lists molecular species and $r\in\Rx$ records that $s(r)$ is consumed and $t(r)$ produced; mathematically,
nothing beyond the finite datum of Definition~\ref{def:petri} is used,
and each chemical term is glossed on first appearance.

Over such a net, the tower is built:
\begin{equation}
\label{eq:fulltower}
\begin{tikzcd}[column sep=1.0em, row sep=1.2em]
\Lk_7(P) \arrow[r, dotted, "\text{BO}"] &
\Lk_6(P) \arrow[r, "U_6"] &
\Lk_5(P) \arrow[r, "U_5"] &
\Lk_4(P) \arrow[r, "U_4^{\Lk_3}"] \arrow[rrrr, bend right=22, "U_4"'] &
\Lk_3(P) \arrow[r, "U_3"] &
\Lk_2(P) \arrow[r, "U_2"] &
\Lk_1(P) \arrow[r, "U_1"] &
\Lk_0(P)
\end{tikzcd}
\end{equation}
The base $\Lk_0(P)$ is the bare network,
the free symmetric monoidal category on $P$ (Section~\ref{sec:L0}).
Each higher level adds one layer of chemical content,
and the labelled arrow back down forgets it.
The lower chain decorates the network with numbers: reaction energetics at $\Lk_1$,
equilibrium at $\Lk_2$, kinetics at $\Lk_3$ (Section~\ref{sec:dict}).

The inclusion of reaction mechanisms, at the level $\Lk_4$,
cannot be understood by decoration as shown in the lower levels up to $\Lk_3$.
Instead, each reaction is replaced by a rewrite of molecular structures,
a double-pushout graph rewrite that records which bonds break and form,
labelled by the reaction it realises (Section~\ref{sec:L4}).
This is the branch point of the diagram.
Two forgetful functors leave $\Lk_4$: one, $U_4$, which lands directly in the base $\Lk_0(P)$;
the other, $U_4^{\Lk_3}$, which drops into the decorated chain.
The two routes to $\Lk_0(P)$ agree, because $U_4$ factors through every decorated level \cite[Prop.~6.31]{Han2026}.

Above $\Lk_4$ the tower turns to the structure of individual molecules.
A short stereochemical half-step $\Lk_{4.5}$ refines the mechanistic level (end of Section~\ref{sec:L4}).
Upon it,
the geometric level $\Lk_5(P)$ assigns to each species a potential energy surface (Section~\ref{sec:L5}),
and the electronic-structure level $\Lk_6(P)$ records the topology of the electronic states above that surface (Section~\ref{sec:L6}).
 
The dotted arrow of Equation \eqref{eq:fulltower} differs from the others:
it is not part of the tower's construction.
It stands for the Born--Oppenheimer separation of electronic and nuclear motion.
Above $\Lk_6$ sits a full quantum level $\Lk_7(P)$,
the description of all-particles in which electrons and nuclei are treated together.
The Born--Oppenheimer approximation is the passage down from $\Lk_7(P)$ to the electronic structure level of $\Lk_6$,
the description the lower tower actually works with.
That top level and the limit that descends from it are left open here (Section~\ref{sec:roadmap}).
 
For levels through $\Lk_0$ to $\Lk_4$, this paper performs the forcing diagnostic itself:
the base automorphism group is computed, and the force of the class for each level is identified as a cokernel or,
at $\Lk_4$, as an inequality of derivations.
Some supporting constructions in $\Lk_4$,
such as the molecular-graph label algebra, are stated rather than reproduced,
with the details in the monograph \cite{Han2026}.
The geometric and electronic levels above have more elaborate target categories.
For these, the paper states each construction and identifies the datum that the level below cannot record,
while importing the detailed proofs from the monograph \cite{Han2026}.
\subsection{The base: the free symmetric monoidal category on a network}
\label{sec:L0}

That a Petri net $P$ freely generates a symmetric monoidal category $\Lk_0(P)$ is due to Meseguer and Montanari \cite{MeseguerMontanari};
the construction is recalled in the skeletal form convenient here,
isolating the universal property, since every later level is built on it.

The choice of structure is forced by the chemistry that it must record.
At the level of pure stoichiometry, before energy, rate, or geometry,
a reaction network carries exactly three structural features.
Reactions \emph{compose in sequence}:
if one reaction produces what another consumes, the two form a composite.
Independent reactions \emph{run in parallel}:
two reactions in disjoint mixtures combine into a single joint process.
A mixture \emph{has no intrinsic order}: 
$A+B$ and $B+A$ are the same physical state, not merely isomorphic ones. 
These three features ---
sequential composition, parallel composition, and order-independence of mixtures
--- 
are precisely the data of a symmetric monoidal category:
composition of morphisms, the tensor product, and the symmetry.
The order-independence is literal equality of states rather than a chosen isomorphism, 
which is why the appropriate form is \emph{strict}:
the object monoid is strictly commutative, 
and $A+B$ and $B+A$ are one object.
Taking the network as generators and imposing no relations beyond those these features demand is the free construction;
assigning numerical or geometric content to a network, at every higher level, 
is then exactly the giving of a functor out of $\Lk_0(P)$.

\begin{definition}[Stoichiometric category]
\label{def:L0}
$\Lk_0(P)$ is the free permutative category on $P$: the strict symmetric monoidal
category, with strictly commutative skeletal object monoid, presented by
\begin{enumerate}[label=(\roman*),leftmargin=2.2em]
\item \textbf{objects} the commutative monoid $\Mon$, with $u\otimes v:=u+v$ and
unit $0$ (so $u\otimes v=v\otimes u$ is literally one object, not an isomorphism
class);
\item \textbf{generating morphisms} the reactions $r\colon s(r)\to t(r)$ for
$r\in\Rx$, together with the symmetries $\sigma_{u,v}$ of the permutative
structure --- which, since $u+v=v+u$ is a single object, are
\emph{endomorphisms} $u+v\to u+v$;
\item \textbf{morphisms} all formal composites and tensors of generators modulo
the permutative congruence (the symmetric strict monoidal axioms), with
\emph{no} relation imposed between distinct reaction labels, even when they share
source and target.
\end{enumerate}
This is functorial in $P$:
a morphism of Petri nets 
(a pair of maps $\Sp\to\Sp'$, $\Rx\to\Rx'$ commuting with $s,t$)
induces a strict symmetric monoidal functor,
so $\Lk_0$ is a functor $\Petri\to\PermCat$,
where $\Petri$ is the category of Petri nets and $\PermCat$ the category of small permutative categories with strict symmetric monoidal functors.
\end{definition}

The clause that distinct labels are never identified (iii)
is the one that makes $\Lk_0$ a faithful record of the network rather than of its net effect: 
two reactions with the same input and output multiset remain different morphisms.
The universal property says these labels, and the species, are the only free data.

\begin{theorem}[Universal property of $\Lk_0(P)$ \cite{MeseguerMontanari}]
\label{thm:univ}
For every strict symmetric monoidal category $C$ with strictly commutative object
monoid, restriction to generators is a bijection
\[
\{\text{strict symmetric monoidal }F\colon\Lk_0(P)\to C\}
\;\xrightarrow{\ \sim\ }\;
\{(g,h)\},
\qquad F\mapsto(F|_{\Sp},\,F|_{\Rx}),
\]
natural in $C$ and in $P$, where $g\colon\Sp\to\Ob(C)$ assigns an object to each
species and $h$ assigns to each $r\in\Rx$ a morphism $\bar g(s(r))\to\bar g(t(r))$,
with $\bar g\colon\Mon\to\Ob(C)$ the monoid homomorphism extending $g$.

Concretely: every such $F$ restricts to the pair $(g,h)=(F|_{\Sp},F|_{\Rx})$, and
every pair $(g,h)$ of the stated form arises from a unique $F$. Naturality in $C$
means a strict symmetric monoidal $\phi\colon C\to C'$ carries the pair of $F$ to the
pair of $\phi\circ F$, by post-composing $g$ and $h$ with $\phi$; naturality in $P$
is the analogous statement for a Petri-net morphism. Thus $\Lk_0(P)$ is free on the
network: the generator data $(g,h)$ are the only choices, and strict functoriality
and monoidality force everything else.
\end{theorem}
 
\begin{proof}
To say $\Lk_0(P)$ is free on $P$ is to say that this restriction map is a
bijection for every $C$ and that the bijection is natural. 
The three ingredients are verified in turn: injectivity, surjectivity, and naturality.

\emph{Injectivity.}
Suppose $F,F'$ restrict to the same $(g,h)$. 
On objects, 
$F$ and $\bar g$ are monoid homomorphisms $\Mon\to\Ob(C)$ agreeing on the free generators $\Sp$, 
hence equal, and likewise $F'=\bar g$; so $F$ and $F'$ agree on objects.
On morphisms they agree on the reaction generators (both equal $h$) and
on the symmetries 
(both forced, $F(\sigma_{u,v})=\sigma^{C}_{Fu,Fv}=F'(\sigma_{u,v})$ by strict symmetry).
Every morphism of $\Lk_0(P)$ is, by the presentation,
a class of formal composites and tensors of these generators;
since $F$ and $F'$ are functors preserving $\circ$ and $\otimes$ and agree on generators,
they agree on all such composites, hence $F=F'$.
 
\emph{Surjectivity.}
Given a pair $(g,h)$ of the stated form, 
define a candidate $F$ as follows.
On objects,
set $F:=\bar g$, the monoid homomorphism extending $g$;
this is forced if $F$ is to be strict monoidal, 
since $\bar g(u+v)=\bar g(u)+\bar g(v)$ and $\bar g(0)=0$.
On the reaction generators set $F(r):=h(r)$, 
which by hypothesis has the required source and target $\bar g(s(r))\to\bar g(t(r))$;
on the symmetry generators set $F(\sigma_{u,v}):=\sigma^{C}_{\bar g(u),\bar g(v)}$;
and on a formal composite or tensor of generators,
let $F$ be the corresponding composite or tensor of their images in $C$.
 
It remains to check that this assignment descends to the morphisms of $\Lk_0(P)$,
which are congruence classes of such formal expressions.
Two expressions are identified in $\Lk_0(P)$ precisely when one is carried to the other by the defining relations:
the bifunctoriality and coherence axioms of a symmetric strict monoidal category,
together with the strict-commutativity identification $u+v=v+u$ on objects.
Each holds in $C$: the coherence axioms because $C$ is strict symmetric monoidal,
and the identification $u+v=v+u$ because the object monoid $\Ob(C)$ is strictly commutative,
which is the hypothesis on $C$. 
Hence congruent expressions receive equal values and $F$ is well defined on $\Lk_0(P)$.
It is moreover a strict symmetric monoidal functor:
it preserves identities and composition because it was defined to commute with $\circ$;
it preserves $\otimes$ and the unit because $\bar g$ is a monoid homomorphism on objects and $F$ commutes with $\otimes$ on morphisms;
and it preserves the symmetry by the definition on $\sigma_{u,v}$.
By construction $F|_{\Sp}=g$ and $F|_{\Rx}=h$, so $(g,h)$ is its restriction.
 
\emph{Naturality.}
For a strict symmetric monoidal $\phi\colon C\to C'$ and $F\colon\Lk_0(P)\to C$,
the pair of $\phi\circ F$ is $(\Ob(\phi)\circ F|_{\Sp},\,\phi\circ F|_{\Rx})$,
which is the pair of $F$ with $g$ and $h$ post-composed by $\phi$;
so restriction commutes with the action of $\phi$, giving naturality in $C$.
For naturality in $P$, a Petri net morphism $p\colon P\to P'$ acts on generators through its species and reaction maps and induces $\Lk_0(p)\colon\Lk_0(P)\to\Lk_0(P')$ by functoriality of $\Lk_0$;
precomposition with $\Lk_0(p)$ on functors corresponds to precomposition with $p$ on generator data,
so restriction commutes with the action of $p$ as well.
The bijection is therefore natural in both arguments,
which is the asserted universal property.
\end{proof}

\begin{remark}[On the hypothesis on $C$]
\label{rem:strictcomm}
The strict-commutativity hypothesis on $C$ is forced by the construction of $\Lk_0(P)$,
whose object monoid is strictly commutative: $u+v$ and $v+u$ are one object,
not two isomorphic ones.
A strict monoidal functor must send that single object to a single object of $C$,
so $\Ob(C)$ must be commutative as well.
The hypothesis is therefore not a restriction in practice.
The targets used to build the higher levels (Section~\ref{sec:dict}) are 
one-object categories such as $\BR$, where strict commutativity holds trivially;
and the target used to compute the automorphisms of $\Lk_0(P)$ (Section~\ref{sec:forcing}) is $\Lk_0(P)$ itself,
whose object monoid is strictly commutative by construction.
In every application the hypothesis is met.
\end{remark}

Theorem~\ref{thm:univ} is used in both directions.
To build a functor \emph{out of} $\Lk_0(P)$,
it suffices to name an object for each species and a morphism for each reaction;
this is how the higher levels are constructed (Section~\ref{sec:dict}).
To constrain the functors \emph{from $\Lk_0(P)$ to itself},
the same bijection pins down every automorphism by its action on generators;
this is how the symmetries of the base are computed (Section~\ref{sec:forcing}).
Both uses lie ahead.
The base also carries numerical invariants that need none of this categorical machinery,
read directly off the network by linear algebra, and these come first.

\paragraph{Stoichiometry, conservation laws, deficiency.}
Order $\Sp$ and $\Rx$ and let $N\in\ZZ^{|\Sp|\times|\Rx|}$ be the
\emph{stoichiometric matrix}, with $r$-th column the net change $t(r)-s(r)$. The
\emph{complexes} are the distinct multisets appearing as some $s(r)$ or $t(r)$;
write $n$ for their number and $\ell$ for the number of connected components
(\emph{linkage classes}) of the directed graph on complexes with an edge
$s(r)\to t(r)$ for each reaction. With $\rho:=\rank N$ the dimension of the
\emph{stoichiometric subspace} $\im N\subseteq\RR^{|\Sp|}$, the
\emph{deficiency} is the integer
\begin{equation}
\delta\;=\;n-\ell-\rho\;\ge\;0 .
\label{eq:deficiency}
\end{equation}
Each term is read off the Petri presentation $P$, so $\delta$ is invariant under
presentation-preserving isomorphisms of $P$; it is not claimed to be invariant
under arbitrary symmetric monoidal equivalence of the underlying category. Both
$\delta$ and the conservation laws below are invariants of the base $\Lk_0(P)$
itself, not data of any higher level: the decorations of Section~\ref{sec:dict}
are functors that extend $\Lk_0(P)$ upward, whereas $\delta$ and the conservation
laws are already determined by $\Lk_0(P)$ alone. A conservation law is an additive
quantity on species unchanged by every reaction: a monoid homomorphism
$w\colon\Mon\to(\RR,+)$ with $w(s(r))=w(t(r))$ for all $r\in\Rx$, equivalently a
covector $w\in\RR^{|\Sp|}$ with $N^{\!\top}w=0$. Such a $w$ is a homomorphism on
the object monoid of $\Lk_0(P)$ that coequalises the source and target maps
$s,t\colon\Rx\rightrightarrows\Mon$; it is therefore \emph{object} data. This is a
second distinction, internal to the base and orthogonal to the first: a
conservation law lives on the objects of $\Lk_0(P)$, while the level-extending
functors of Section~\ref{sec:dict} assign their data to its \emph{morphisms}. What
\eqref{eq:deficiency} does \emph{not} determine is any rate or equilibrium: those
are the data of the higher levels, and the interaction between $\delta$ and the
$\Lk_3$ datum is the content of the deficiency theorems, treated in the
application section.
\subsection{Energetics, equilibrium, kinetics as functors}
\label{sec:dict}

The three lower levels share a form:
adjoin to the fixed base a strict symmetric monoidal functor into a one-object target,
with the forgetful operation $U_k$ discarding that functor and returning the bare base.
By Theorem~\ref{thm:univ} each such functor is freely specified by one value per reaction, 
so a level of this kind is a labelling of reactions by elements of an algebraic structure, 
with composites sent to products.
What distinguishes the levels is which structure,
and what its functor axioms then assert about chemistry.

\subsubsection{$\Lk_1$: energetics, additivity, and Hess path-independence}
\label{sec:L1}

\paragraph{$\Lk_0$ cannot see enthalpy.}
Hess's law, the oldest quantitative law of thermochemistry, 
is usually stated as constant heat summation:
the enthalpy change of a reaction is the sum of the enthalpy changes of any sequence of steps that accomplishes it,
\[
  \Delta H=\Delta H_1+\Delta H_2+\cdots,
\]
and is the same whatever sequence is chosen.
As ordinarily taught, the law fuses two claims:
that enthalpies \emph{add along a chosen route},
and that the total is \emph{independent of the route},
the latter being the statement that enthalpy is a state function.

The failure of bare stoichiometry appears already in carbon combustion:
\[
  \ch{C(graphite) + O2(g) -> CO2(g)}
  \qquad\text{and}\qquad
  \ch{C(diamond) + O2(g) -> CO2(g)} .
\]
At $\Lk_0$ both have the same stoichiometric shadow,
\(\ch{C + O2 -> CO2}\),
yet they have different heats of combustion.
The $\Lk_1$ construction is forced precisely by this gap:
it realises additivity along a chosen route as functoriality of the enthalpy functor, $F_H$,
while isolating route-independence as a separate state-function condition. 
The present tower therefore records which of the two claims any given use of Hess's law actually requires.

\paragraph{The new datum: an additive functor $\FH$.}
Let $\BR$ be the one-object strict SMC with object $\ast$,
endomorphism monoid $(\RR,+)$ under composition (so $b\circ a=a+b$, addition being
commutative the order is immaterial, with identity $0$), tensor also $+$, and
symmetry $\sigma_{\ast,\ast}=0$. A strict symmetric monoidal functor
$\Lk_0(P)\to\BR$ is, by Theorem~\ref{thm:univ}, exactly a function $\Rx\to\RR$
sending composites and tensors to sums.

\begin{definition}[Energetic level]
\label{def:L1}
\[
\Lk_1(P):=(\Lk_0(P),\FH),
\]
where $\FH\colon\Lk_0(P)\to\BR$ is a strict symmetric monoidal functor. The
forgetful functor $U_1\colon\Lk_1(P)\to\Lk_0(P)$ is the operation that forgets
$\FH$.
\end{definition}

By Theorem~\ref{thm:univ}, such an $\FH$ is the same datum as a function
$\eta\colon\Rx\to\RR$, namely $\eta=\FH|_{\Rx}$, with composites and tensors sent to
sums:
\[
\FH(r_2\circ r_1)=\FH(r_1)+\FH(r_2),\qquad
\FH(r_1\otimes r_2)=\FH(r_1)+\FH(r_2),\qquad
\FH(\id_u)=0 .
\]

\begin{proposition}
\label{prop:FHfree}
For every $\eta\colon\Rx\to\RR$ there is a unique
$\FH\colon\Lk_0(P)\to\BR$ with $\FH|_{\Rx}=\eta$.
\end{proposition}
\begin{proof}
Theorem~\ref{thm:univ} with $C=\BR$: the object assignment is forced
(one object) and the morphism assignment is the free datum $\eta$, with no
source--target constraint to satisfy since $\BR$ has a single object.
\end{proof}

\begin{remark}[Chemical reading]
\label{rem:L1chem}
The value $\FH(r)\in\RR$ is the \emph{enthalpy of reaction} of $r$, signed in the
standard convention (negative for an exothermic reaction, positive for an
endothermic one). The composition law $\FH(r_2\circ r_1)=\FH(r_1)+\FH(r_2)$ is the
additive half of Hess's law: along any chosen composite of reactions, the assigned
enthalpies add. This holds for every $\FH$, because $\FH$ is a functor.
\end{remark}

\paragraph{The coboundary condition and path-independence.}
Functoriality (Remark~\ref{rem:L1chem}) is weaker than the condition that the
chemistry of $\Lk_1$ ultimately turns on. Say that $\FH$ \emph{is a coboundary} if
there is a monoid homomorphism $h\colon\Mon\to(\RR,+)$, the \emph{object potential},
with
\begin{equation}
\FH(r)=h(t(r))-h(s(r)),\qquad r\in\Rx .
\label{eq:statefn}
\end{equation}
This is a genuinely stronger condition on $\FH$ than functoriality: not every strict
symmetric monoidal $\FH\colon\Lk_0(P)\to\BR$ arises from an object potential, and the
ones that do are exactly the coboundaries in the sense of \eqref{eq:statefn}. When
$\FH$ is a coboundary, every morphism $u\to v$ receives the value $h(v)-h(u)$
independently of the route from $u$ to $v$; and \eqref{eq:statefn} forces
$\FH(r_1)=\FH(r_2)$ whenever $s(r_1)=s(r_2)$ and $t(r_1)=t(r_2)$, though it does not
identify $r_1$ and $r_2$ as morphisms. The generic functor $\FH$ of
Definition~\ref{def:L1} is the categorical content of $\Lk_1$; the coboundary
condition \eqref{eq:statefn} is a separate, strictly stronger property, kept apart
from the definition throughout, following \cite[Rem.~2.41]{Han2026}.

\begin{remark}[Chemical reading of the coboundary condition]
\label{rem:L1coboundary}
Chemically, an object potential $h$ is an assignment of an enthalpy of formation to
each species, and \eqref{eq:statefn} is the statement that the enthalpy of reaction
is the difference of formation enthalpies between products and reactants. The
route-independence that \eqref{eq:statefn} then delivers --- the enthalpy of a
transformation $u\to v$ depends only on $u$ and $v$, not on the path --- is the
second, stronger half of Hess's law, the statement that enthalpy is a state
function. The tower thus separates the two claims ordinarily fused under Hess's law:
plain additivity is functoriality of $\FH$, while path-independence is the coboundary
condition \eqref{eq:statefn}, and a given use of the law may require only the former.
\end{remark}

\paragraph{What $\Lk_1$ adds to $\Lk_0$.}
This is also what $\Lk_1$ adds to the base.
At $\Lk_0(P)$ two reactions with the same source and target are interchangeable:
nothing in the bare bookkeeping of species consumed and produced distinguishes them,
so relabelling one as the other is a symmetry of the base.
The enthalpy functor breaks that symmetry,
since the two reactions may carry different values of $\FH$.
The datum that $\Lk_1$ records is exactly this lost interchangeability:
enthalpy is a distinction $\Lk_0(P)$ cannot make,
and the symmetries of $\Lk_0(P)$ that fail to lift to $\Lk_1(P)$ are computed as a cokernel in Section~\ref{sec:forcing}. 
Hess additivity is, in this reading, the first new datum above stoichiometry.
\subsubsection{$\Lk_2$: equilibrium, reversal symmetry, and Wegscheider cycles}
\label{sec:L2}

\paragraph{$\Lk_1$ cannot see entropy and free energy.}
Enthalpy does not determine equilibrium.
Equilibrium is set by the Gibbs free energy $\Delta G=\Delta H-T\Delta S$,
which carries an entropic term alongside the enthalpic one,
and the two contributions are independent:
a reaction may be driven by a favourable enthalpy, 
by a favourable entropy, or by one against the other.
Knowing $\Delta H$ alone fixes neither the position of equilibrium nor how it moves with temperature,
since the entropic contribution $T\Delta S$ is free to vary.
Isothermal titration calorimetry of protein--ligand binding makes the point concretely:
it resolves $\Delta H$ and $\Delta S$ separately,
and across a ligand series binding is found to be enthalpy-driven in some cases
and entropy-driven in others at comparable affinity,
so two reactions can share a value of $\Delta H$ and yet differ in free energy and in temperature response.
The level $\Lk_1$, which records only $\FH$, cannot tell such reactions apart;
reaching equilibrium requires a second datum, the entropy,
and the temperature-weighted combination of the two.

\paragraph{The new datum: an entropy functor $\FS$ and the Gibbs functor.}
By Theorem~\ref{thm:univ}, that second datum is again freely specified: a strict
symmetric monoidal functor $\FS\colon\Lk_0(P)\to\BR$ is exactly one real value per
reaction, with composites and tensors sent to sums. Combining it with $\FH$ under a
temperature weight gives the one-parameter family defined below. Equilibrium is the
organising question of this level, and it draws on three pieces of structure in
turn: the entropy functor and the free energy built from it, a reversal symmetry on
reversible networks, and the consistency of free energy around cycles. The section
takes them in that order.

\begin{definition}[Equilibrium level]
\label{def:L2}
\[
\Lk_2(P):=(\Lk_0(P),\FH,\FS),
\]
where $\FH,\FS\colon\Lk_0(P)\to\BR$ are strict symmetric monoidal functors. The
forgetful functor $U_2\colon\Lk_2(P)\to\Lk_1(P)$ is the operation that forgets
$\FS$.
\end{definition}

\begin{definition}[Gibbs functor]
\label{def:gibbs}
For each $T>0$ the \emph{Gibbs functor} of an equilibrium level $(\Lk_0(P),\FH,\FS)$
is the strict symmetric monoidal functor
\[
\FG^{T}:=\FH-T\,\FS\;\colon\;\Lk_0(P)\to\BR ,
\]
the pointwise difference taken in the endomorphism monoid $(\RR,+)$ of $\BR$.
\end{definition}

By Theorem~\ref{thm:univ}, $\FS$ is the same datum as a function $\Rx\to\RR$, one
real value per reaction, with composites and tensors sent to sums. A single
$\FG^{T_0}$ retains only the difference $\FH(r)-T_0\FS(r)$; the family
$\{\FG^T\}_{T>0}$, affine in $T$ per reaction, is interchangeable with the pair
$(\FH,\FS)$.

\begin{remark}[Chemical reading]
\label{rem:L2chem}
The value $\FS(r)\in\RR$ is the \emph{entropy change} of $r$, a measure of how the
reaction changes the dispersal of energy among states, and $\FG^T(r)$ is its Gibbs
free energy of reaction at temperature $T$. The equilibrium constant of $r$ is
recovered as $K_r(T)=\exp(-\FG^T(r)/RT)$, and its temperature dependence (the van~'t
Hoff relation, in the usual approximation where $\FH$ and $\FS$ are treated as
temperature-independent over the range considered) is the $T$-derivative of
$\FG^T$; both are thus functions of the decoration, internal to $\Lk_2$.
\end{remark}

\paragraph{Reversal symmetry on reversible networks.}
A new piece of structure appears at this level that has no counterpart at $\Lk_1$.
When $P$ is \emph{reversible} --- 
closed under a fixed-point-free involution $r\mapsto r^{\dagger}$ with $s(r^{\dagger})=t(r)$, $t(r^{\dagger})=s(r)$
--- the assignment $r\mapsto r^{\dagger}$ extends,
by Theorem~\ref{thm:univ},
to a strict symmetric monoidal \emph{dagger} on $\Lk_0(P)$
(a contravariant involutive identity-on-objects functor).
The decorations are compatible with it precisely when $\FH(r^{\dagger})=-\FH(r)$ and $\FS(r^{\dagger})=-\FS(r)$,
hence $\FG^T(r^{\dagger})=-\FG^T(r)$.

\begin{proposition}[Reversal symmetry]
\label{prop:reversal}
For a reversible $P$ with $\dagger$-compatible Gibbs functor,
$\FG^T(r^{\dagger})=-\FG^T(r)$ for every reaction,
equivalently $K_{r^{\dagger}}(T)=K_r(T)^{-1}$ for the equilibrium constants.
\end{proposition}
\begin{proof}
Immediate from $\dagger$-compatibility and $\FG^T(r)=-RT\log K_r(T)$.
\end{proof}

\paragraph{Cycle consistency of free energy.}
Reversal symmetry constrains each forward--reverse pair but says nothing about longer cycles.
The cycle conditions --- 
consistency of the free-energy increments
(equivalently the equilibrium constants)
around an arbitrary loop ---
are a strictly stronger requirement,
and they hold precisely when the Gibbs functor is a coboundary,
exactly as for the enthalpy in \eqref{eq:statefn}.
The companion \emph{kinetic} detailed-balance condition relating rate constants,
$k_r/k_{r^{\dagger}}=K_r(T)$,
is imposed separately at $\Lk_3$.

\begin{proposition}[Wegscheider cycle conditions]
\label{prop:wegscheider}
Suppose $\FG^T=g(t(-))-g(s(-))$ for an object potential $g\colon\Mon\to(\RR,+)$
(the state-function condition for free energy).
Then for every directed cycle
$u_0\xrightarrow{r_1}u_1\xrightarrow{r_2}\cdots\xrightarrow{r_m}u_0$,
$$
\sum_{i=1}^{m}\FG^T(r_i)=0,\qquad\text{equivalently}\qquad
\prod_{i=1}^{m}K_{r_i}(T)=1 .
$$
\end{proposition}
\begin{proof}
Along the cycle each reaction has $s(r_i)=u_{i-1}$ and $t(r_i)=u_i$,
so the hypothesis gives $\FG^T(r_i)=g(u_i)-g(u_{i-1})$ for each $i$.
Summing over the cycle,
\[
\sum_{i=1}^{m}\FG^T(r_i)=\sum_{i=1}^{m}\bigl(g(u_i)-g(u_{i-1})\bigr)
=g(u_m)-g(u_0)=0,
\]
the intermediate terms cancelling in pairs and $u_m=u_0$ closing the cycle.
Exponentiating $\FG^T(r)=-RT\log K_r(T)$ turns this additive identity into the multiplicative one $\prod_i K_{r_i}(T)=1$.
\end{proof}

\paragraph{Detailed balance as two separate conditions.}
The division of labour is the point.
The $\dagger$-structure (Proposition~\ref{prop:reversal})
is genuinely new at $\Lk_2$ and handles reversal;
the cycle conditions (Proposition~\ref{prop:wegscheider}) are
not a consequence of $\dagger$-compatibility alone 
but require the same coboundary condition that upgraded additivity to path-independence at $\Lk_1$.
Detailed balance is thus two conditions,
not one, and the tower keeps them visibly separate.

\paragraph{What $\Lk_2$ adds to $\Lk_1$.}
What $\Lk_2$ adds to $\Lk_1$ is the resolution of the gap that opened this section.
Take two reactions $r_1,r_2\colon u\to v$ with $\FH(r_1)=\FH(r_2)$ but $\FS(r_1)\neq\FS(r_2)$.
The label swap $r_1\leftrightarrow r_2$ preserves the $\Lk_1$ decoration,
since it leaves $\FH$ unchanged on the two reactions,
so it is a symmetry of $(\Lk_0(P),\FH)$;
but it does not preserve $\FS$, so it is not a symmetry of $\Lk_2(P)$.
The entropy functor breaks it.
This lost symmetry is the datum $\Lk_2$ records:
the underlying $\Lk_1$ decoration does not determine the entropy,
the swap represents a non-trivial class in the cokernel computed in Section~\ref{sec:forcing},
and the two reactions that $\Lk_1$ conflates are separated at every temperature by $\FG^T(r_1)\neq\FG^T(r_2)$.
\subsubsection{$\Lk_3$: kinetics as an additive channel decoration}
\label{sec:L3}

\paragraph{$\Lk_2$ cannot see the reaction rate.}
Chemical equilibrium fixes where a reaction ends up, but not how fast it gets there.
Two reactions with identical thermodynamics, the same $\Delta G$ and hence the same
equilibrium constant, can proceed at rates differing by many orders of magnitude:
the hydration of $\ch{CO2}$ runs about $10^{7}$ times faster in the presence of carbonic anhydrase than without it, 
with the catalyst leaving the equilibrium position untouched.
Rate is therefore a datum that $\Lk_2$ cannot record; it enters at $\Lk_3$.
Each rate constant absorbs physical conditions such as temperature, pressure, and solvent, 
which lie outside the stoichiometry $(\Sp,\Rx,s,t)$ altogether;
and even when a catalyst is itself a species of the network,
the magnitude of the rate enhancement it confers is a property of the chemistry,
not of the combinatorics.
The rate constant must therefore be supplied as an empirical input rather than derived from the network.
Kinetics records the rate of each reaction and the resulting stochastic evolution of molecule counts.

\paragraph{The rate datum: mass-action propensities.}
The rate of a reaction is set by the \emph{law of mass action}:
a reaction fires at a rate equal to its rate constant times the number of distinct ways its
reactant molecules can be selected from those present. 
In a state with a copy-number vector $x\in\Mon$,
this makes the rate of reaction $r$ the rate constant $k_r$ times the product, over species,
of the number of ways to choose $s(r)_S$ molecules of species $S$ from the $x_S$ available.
Each firing consumes $s(r)$ and produces $t(r)$, changing the state by $n_r=t(r)-s(r)$. 
This rate, the \emph{propensity} $\lambda_r(x)$, is what kinetics adds to the network.

\paragraph{The new datum: channel generators and the target $\Bg$.}
To make this a decoration of the same form as the others,
the values must be composed additively, but propensities do not:
a sequence of two reactions does not have a rate equal to the sum of the two rates.
What does add is the infinitesimal generator of the Markov dynamics that each reaction induces.
The kinetic functor therefore should take values in generators rather than rates,
assigning to each reaction its channel generator and to a formal path the sum of the generators along it.
Summing along a path is bookkeeping internal to the functor;
it is not the physical combination of reactions,
which is the network-wide aggregate $\Omega_P$ defined below.
The two are kept distinct: a categorical path is a record of which reactions occur,
not a model of their concurrent action.

Fix the population space $\RR^{\Mon}$ of functions $f\colon\Mon\to\RR$.
To each reaction $r$ associate the \emph{channel generator}
$$
\bigl(\Omega_r f\bigr)(x)=\lambda_r(x)\,\bigl[f(x+n_r)-f(x)\bigr],
\quad n_r=t(r)-s(r),\quad
\lambda_r(x)=k_r\!\prod_{S\in\Sp}\binom{x_S}{s(r)_S},\ \ k_r>0,
$$
an operator on $\RR^{\Mon}$,
with $\lambda_r$ the mass-action propensity just described.
The binomial coefficient vanishes when $x_S<s(r)_S$ for any species,
so $\lambda_r(x)=0$ automatically unless $x\ge s(r)$ componentwise:
a reaction cannot fire when its reactants are not present.
Wherever $\lambda_r(x)\neq0$ the reached state $x+n_r=x-s(r)+t(r)$ is again in $\Mon$,
so $\Omega_r$ never evaluates $f$ at a nonphysical state.
Everything is taken algebraically on finitely supported functions,
so that no analytic domain or boundedness statement is intended.
Let $\gen$ be the commutative monoid generated by $\{\Omega_r\}$ under operator addition,
with unit the zero operator,
and let $\Bg$ be the one-object strict SMC with endomorphism monoid $\gen$,
both composition and tensor given by this addition,
and symmetry the identity $\sigma=0$.
The construction intentionally quotients by equality of generator contributions:
if two reactions induce the same operator $\Omega_r$ they are identified in $\gen$,
so $\FP$ is not faithful.
This is by design: $\FP$ records kinetic behaviour as an
observable, while the reaction labels are kept by $\Lk_0(P)$,
which $\FP$ decorates.
It deliberately does \emph{not} use a Kronecker-sum tensor;
that would require a typed category whose objects are distinct population spaces $\RR^{\NN[S']}$,
the open-network infrastructure outside the present scope.
This $\Bg$ is the fixed-network additive target used for the present compact theory section,
not a typed monoidal semantics for open networks.

\paragraph{The kinetic level and its master equation.}
With $\Bg$ as target,
the kinetic decoration is again freely specified in the sense of Theorem~\ref{thm:univ}:
a strict symmetric monoidal functor $\Lk_0(P)\to\Bg$ is exactly one channel generator per reaction.
The level $\Lk_3$ takes that datum to be the mass-action generator $\Omega_r$ of each reaction.

\begin{definition}[Kinetic level]
\label{def:L3}
$\Lk_3(P):=(\Lk_0(P),\FH,\FS,\FP)$, where $\FP\colon\Lk_0(P)\to\Bg$ is the strict
symmetric monoidal functor determined (Theorem~\ref{thm:univ}) by $\FP(r)=\Omega_r$
on generators. The forgetful functor $U_3\colon\Lk_3(P)\to\Lk_2(P)$ is the operation
that forgets $\FP$.
\end{definition}

The functor records additive bookkeeping of channel contributions:
a formal path $u\xrightarrow{r_1}v\xrightarrow{r_2}w$ is sent to $\Omega_{r_1}+\Omega_{r_2}$,
and a tensor likewise to a sum.
This is all $\FP$ asserts, and it must be read with care,
because categorical composition in $\Lk_0(P)$ is a \emph{formal path} of reactions, not concurrency:
\begin{itemize}[leftmargin=1.4em]
\item $\FP$ is an \emph{additive decoration}:
it assigns to every morphism the sum of the channel generators occurring in it.
It does \emph{not} assign to a composite the effective generator of a coarse-grained single step,
which would require a quasi-steady-state elimination and is not a functor operation.
\item The object modelling the network's stochastic dynamics is the \emph{aggregate} over all channels,
$$
\Omega_P=\sum_{r\in\Rx}\Omega_r ,
$$
a network-level sum over the channels available in parallel,
the image of the whole network under the semantic map $P\mapsto\Omega_P$,
not the image of any path composite under $\FP$.
\end{itemize}

The aggregate $\Omega_P$ is the generator of the \emph{chemical master equation},
the equation that governs the stochastic dynamics of the network.
A state of the system is a copy-number vector $x\in\Mon$ recording
how many molecules of each species are present,
and the master equation evolves a probability distribution $p_t$ over these states,
\begin{equation}
\frac{d}{dt}\,p_t=\Omega_P\,p_t,\qquad
p_t=e^{t\Omega_P}p_0 ,
\label{eq:cme}
\end{equation}
with $e^{t\Omega_P}$ the Markov semigroup it generates.
Each channel generator $\Omega_r$ contributes the gain-loss balance of a single reaction:
at rate $\lambda_r(x)$ the reaction fires, moving probability from state $x$ to state $x+n_r$.
Summing the channels gives the full generator,
so that $\Omega_P$ is the exact stochastic law of the network at finite molecule numbers,
and the kinetic datum $\FP$ is precisely what \eqref{eq:cme} requires beyond stoichiometry.

\paragraph{The deterministic limit and chemical reaction network theory.}
The master equation \eqref{eq:cme} is the stochastic, finite-count description of a reaction network.
Chemical reaction network theory more often works with its deterministic counterpart,
the reaction rate equation on concentrations,
and the two are related by a precise scaling limit rather than an informal approximation.
Index the system by a volume $V>0$ and pass from integer copy-number states $x\in\Mon$ to real concentrations $c=x/V$.
As $V$ grows the rate constants are rescaled by the standard convention
that keeps each reaction's rate finite in the limit,
so that for each $V$ the network has a master-equation generator $\Omega_P^{V}$,
the volume-$V$ form of $\Omega_P$.

What these generators converge to was settled by Kurtz \cite{Kurtz1970,Kurtz1972,AndersonCraciunKurtz}, 
whose large-volume limit theorem is 
the precise sense in which the stochastic description reduces to the deterministic one.
As $V\to\infty$, the volume-scaled mass-action generators $\Omega_P^{V}$ converge,
on smooth observables of the concentration $c$,
to the first-order operator $f\mapsto\nabla f\cdot N\,v(c)$,
the generator of the deterministic flow with vector field $c\mapsto N\,v(c)$.
Here $N$ is the stoichiometric matrix of Section~\ref{sec:L0} and
$$
v(c)=\bigl(k_r\,c^{\,s(r)}\bigr)_{r\in\Rx}
$$
is the vector of mass-action rates, 
the rate constants $k_r$ being the large-volume rescaling of the stochastic constants carried by $\FP$
and the exponents the source stoichiometries $s(r)$.
Kurtz's theorem gives more than this operator convergence:
with $c(0)=x_V(0)/V$ held fixed, the stochastic process governed by $\Omega_P^{V}$ converges in probability, 
uniformly on compact time intervals,
to the solution of the \emph{reaction rate equation}
\begin{equation}
\dot c = N\,v(c),
\label{eq:rre}
\end{equation}
so the random trajectories at finite volume concentrate on the deterministic flow as the volume grows.
The tower form of this statement is given in \cite[Prop.~5.18]{Han2026}.
 
This limit is relevant here because it is a forgetful passage out of $\Lk_3$, and
it is what keeps the deterministic theory inside the tower rather than beside it.
The right-hand side $N\,v(c)$ of \eqref{eq:rre} is assembled from data the tower already carries:
$N$ is the $\Lk_0$ stoichiometric matrix and $v$ is built from the rate constants of the $\Lk_3$ functor $\FP$.
The rate equation is therefore not an independent postulate but the large-volume image of $\FP$,
so that $\Lk_3$ contains the deterministic description of $\Lk_0$ as a limiting case,
and no structure built so far is disturbed:
\eqref{eq:rre} is a derived shadow of the stochastic generator $\Omega_P$,
obtained as a forgetful functor on the volume-indexed families \cite[\S5.4.1]{Han2026}.

Equation~\eqref{eq:rre} is the central object of chemical reaction network theory,
and it is exactly the rate equation that Baez and Pollard construct functorially
in their compositional framework \cite{BaezPollard}: 
a reaction network, presented as a Petri net, is sent to its rate equation by a symmetric monoidal functor,
and open networks glued along shared species have rate equations that compose accordingly.
That framework studies the deterministic object \eqref{eq:rre} directly and in the open setting,
where networks are composed along shared species;
the present level supplies the stochastic generator above it and,
by the limit just described,
recovers \eqref{eq:rre} as its deterministic shadow.
The compositional gluing of open networks is a separate development from the fixed-network theory recorded here.

Once the rate datum is present,
two classical results of the theory become available.
The Feinberg--Horn--Jackson deficiency-zero theorem gives global stability of complex-balanced equilibria in weakly reversible deficiency-zero networks,
and its stochastic counterpart, the Anderson--Craciun--Kurtz theorem,
gives product-form Poisson stationary distributions.
Both require the $\Lk_3$ rate data and not the $\Lk_0$ deficiency $\delta$ alone,
and both are treated in the application section.

\paragraph{What $\Lk_3$ adds to $\Lk_2$.}
What $\Lk_3$ adds to $\Lk_2$ is the rate that opened this section.
Take the two hydration pathways for $\ch{CO2}$, $r_A$ uncatalysed and $r_B$ enzyme-catalysed,
with the same source and target and hence the same $\Delta G$.
The swap $r_A\leftrightarrow r_B$ preserves every $\Lk_2$ datum,
since the two share enthalpy, entropy, and equilibrium constant, 
so it is a symmetry of $\Lk_2(P)$;
but $\FP(r_A)\neq\FP(r_B)$,
their channel generators carrying rate constants that differ by orders of magnitude,
so the swap is not a symmetry of $\Lk_3(P)$.
The kinetic functor breaks it.
This lost symmetry is the datum $\Lk_3$ records:
the thermodynamic decoration does not determine the rate,
and the swap represents a non-trivial class in the cokernel computed in Section~\ref{sec:forcing}.
The two reactions that $\Lk_2$ conflates are separated at $\Lk_3$ by their generators.

\begin{remark}[Generators, not kernels]
\label{rem:generators}
One might expect the kinetic target to be a category of Markov kernels.
It cannot be, functorially: 
passing from a generator $L$ to its finite-time kernel $e^{tL}$ does not respect composition,
since $e^{t(L_1+L_2)}\ne e^{tL_1}e^{tL_2}$ unless $[L_1,L_2]=0$.
The additive structure that makes $\FP$ a functor lives at the level of generators;
the exponential to kernels is available pointwise but is not itself a symmetric monoidal functor.
This non-commutativity is a warning that additive generator bookkeeping has forgotten sequential structure;
the next level supplies one chemically meaningful source of such structure (mechanism), 
though it does not by itself resolve the analytic non-commutativity of Markov generators \cite[\S5,\S6]{Han2026}.
\end{remark}

\paragraph{Summary of the the lower levels of the tower.}
Table~\ref{tab:dict} collects the correspondence.
The point is not that any one entry is surprising,
since each is a familiar fact,
but that they are uniformly the same kind of object:
functorial data on one free category.
Reading the table this way separates classical laws that are usually stated together.
The additive laws (Hess's law, reaction-heat bookkeeping) are plain functoriality,
holding for every such functor.
Path-independence and the Wegscheider cycle conditions are stronger:
they hold only when the functor is a coboundary,
that is, when its labelling comes from a state function on species rather than from arbitrary reaction values.
Reversal symmetry is different again, the separate contribution of the
$\dagger$-structure on reversible nets. The tower assigns these three to distinct
conditions rather than bundling them, so each result records exactly which datum it
depends on. Table~\ref{tab:dict} also lists the three levels above kinetics:
mechanism ($\Lk_4$), geometry ($\Lk_5$), and electronic structure ($\Lk_6$). These
change the construction in kind rather than adding a further functor of the present
form, and they are developed in Sections~\ref{sec:L4}--\ref{sec:L6}; the table
records their place in the dictionary in advance.
 
\begin{table}[t]
\centering\small
\renewcommand{\arraystretch}{1.3}
\begin{tabular}{@{}l >{\raggedright\arraybackslash}p{4.3cm}
                       >{\raggedright\arraybackslash}p{6.2cm}@{}}
\toprule
Level & Categorical datum & Classical content \\
\midrule
$\Lk_0$ & free SMC $\Lk_0(P)$; matrix $N$ & stoichiometry, $\delta$;
          conservation laws as object potentials \\
$\Lk_1$ & functor $\FH\colon\Lk_0(P)\to\BR$ & additive reaction heats;
          Hess path-independence when $\FH$ is a coboundary \\
$\Lk_2$ & functors $\FS,\FG^T$; $\dagger$ on reversible nets & equilibrium
          constants; reversal symmetry; Wegscheider cycles when $\FG^T$ is a
          coboundary \\
$\Lk_3$ & functor $\FP\colon\Lk_0(P)\to\Bg$; aggregate $\Omega_P$ & mass-action
          CME; kinetic detailed balance as $\Lk_2$-compatibility \\
$\Lk_4$ & free SMC on labelled DPO rules; $U_4\colon\Lk_4(P)\to\Lk_0(P)$ &
          mechanism; one effective channel, non-isomorphic derivations
          (concerted vs.\ stepwise) \\
$\Lk_5$ & functor $\FV$ into $\OrbMorse$; $\Lk_5(P)=(\Lk_{4.5}(P),\FV)$ &
          geometry; a potential energy surface on the configuration orbifold
          $\Ce(G)$ for each species \\
$\Lk_6$ & electronic bundle over $\Ce(G)$, target $\HilbBundR$ & electronic
          structure; Berry sign holonomy $\eta_B=w_1(\Lzero^{\mathbb R})$ around a
          conical intersection \\
\bottomrule
\end{tabular}
\caption{The tower as a functorial dictionary. Levels $\Lk_0$--$\Lk_3$ each add one
functor out of the fixed base $\Lk_0(P)$, with $U_k$ removing the last; $\Lk_4$
instead enlarges the morphisms, and $\Lk_5$--$\Lk_6$ pass to the geometry and
electronic structure of individual species. The decorated rows $\Lk_1$--$\Lk_3$
are constructed in this section; $\Lk_4$ is treated in Section~\ref{sec:L4} and
$\Lk_5$--$\Lk_6$ in Sections~\ref{sec:L5}--\ref{sec:L6}. The non-collapse of the
decorated sub-tower is Theorem~\ref{thm:autoL0} and Corollary~\ref{cor:stab}; the
separation of $\Lk_4$ from the lower observables is
Proposition~\ref{prop:Dmech}.}
\label{tab:dict}
\end{table}
\subsection{Why the levels do not collapse}
\label{sec:forcing}
When does adjoining a decoration genuinely refine a level,
rather than leave it unchanged?
A decoration refines a level exactly when it separates two reactions
that could not be distinguished in the level below.
A broken symmetry is what certifies a real separation.
A symmetry of $\Lk_0(P)$ that swaps two reactions certifies that no observable below can distinguish them.
A decoration that distinguishes them is then exactly one that breaks the symmetry.
The question therefore becomes one about the symmetries a decoration leaves intact.
These form its stabiliser:
the subgroup of base automorphisms $\sigma$ with $D\circ\sigma=D$,
the symmetries that survive the decoration $D$.
Compute the automorphism group of $\Lk_0(P)$,
then, for each decoration, find this stabiliser.
The outcome depends on the network.
If the network carries a \emph{forcing pair},
two parallel reactions with distinct heats, or distinct entropies, or distinct rates,
the levels are strictly separated.
If it is degenerate and carries no such pair, a level may add no new distinction.
The test below makes the difference explicit.

Consider a self-equivalence of a level:
a strict symmetric monoidal functor from the level to itself with an inverse.
The forcing question needs such a map to act as a relabelling of the chemistry,
carrying species to species and reactions to reactions.
An arbitrary self-equivalence need not do this;
it may rearrange composites and tensors in ways that answer to no relabelling of the network,
and the diagnostic below would say nothing about it.
The relevant maps are therefore those that fix the chemical generators setwise,
the species $\Sp$ among objects and the reactions $\Rx$ among morphisms.
These are the \emph{presentation-preserving automorphisms} of the level;
they form a group $\AutP(-)$.
It matters what such an automorphism certifies,
because the certificate runs in only one direction.
A functor $F$ out of a category need not satisfy $F\circ\sigma=F$ for an automorphism $\sigma$ of its source,
since one can always label two parallel reactions differently.
What is true is a statement about the decorations already chosen.
Suppose $\sigma$ swaps two reactions $r_1,r_2$ and lies in the stabiliser of a decoration $D$.
Then $D$ does not distinguish $r_1$ from $r_2$.
A swapping automorphism is therefore a sufficient diagnostic for indistinguishability,
not a characterisation of it.
This is all the forcing argument needs.
The question at each level is which base automorphisms survive the decorations adjoined so far,
where to survive means to remain compatible with them.
Indistinguishability at level $k$ means indistinguishability by the level-$k$ observables,
not by every conceivable functor.

To make this quantitative,
call the \emph{level-$k$ profile} of a reaction the tuple of data the level records.
At $\Lk_3$ this is the source, target, enthalpy, entropy, and kinetic generator contribution $(s,t,\FH,\FS,\FP)$; 
at $\Lk_1,\Lk_2$ it is the corresponding shorter tuple.
Two reactions have the \emph{same level-$k$ profile} when these agree;
they remain distinct morphisms of the syntax throughout.
The base symmetries that survive level $k$ are exactly those preserving the level-$k$ profile,
and computing them is computing how much each decoration sees.

For the base,
the universal property gives the full symmetry group outright.
Two families of automorphisms are visible.
First, a permutation $\pi$ of species that preserves the \emph{typed reaction multiset} ---
meaning the induced map on $\Mon$ satisfies 
$|\Rx_{u,v}|=|\Rx_{\pi u,\pi v}|$ for every pair of complexes $(u,v)$,
so that fibres are carried to fibres of equal size ---
extends, by Theorem~\ref{thm:univ} 
(after choosing, for each pair, a bijection $\Rx_{u,v}\to\Rx_{\pi u,\pi v}$),
to an automorphism of $\Lk_0(P)$; these form the subgroup
\[
\Symm(\Sp)_P:=\bigl\{\pi\in\Symm(\Sp): |\Rx_{u,v}|=|\Rx_{\pi u,\pi v}|
\text{ for all } u,v\bigr\}\ \le\ \Symm(\Sp).
\]
(Merely preserving the \emph{support} $\{(s(r),t(r))\}$ is not enough: a species
permutation could then send a two-element fibre to a one-element fibre, with no
bijection of labels above it.) Second, for each pair of complexes $(u,v)$, any
permutation of the \emph{fibre}
\[
\Rx_{u,v}:=(s,t)^{-1}(u,v)=\{r\in\Rx: s(r)=u,\ t(r)=v\},
\]
a relabelling of reactions that share source and target, 
extends to an automorphism fixing all species.

\begin{theorem}[Automorphisms of the base]
\label{thm:autoL0}
There is a short exact sequence of groups
\[
1\;\to\;\textstyle\prod_{(u,v)}\Symm(\Rx_{u,v})\;\to\;
\AutP\bigl(\Lk_0(P)\bigr)\;\to\;\Symm(\Sp)_P\;\to\;1 ,
\]
where the normal factor $K:=\prod_{(u,v)}\Symm(\Rx_{u,v})$ consists of the automorphisms fixing every species,
and $\Symm(\Sp)_P$ is the quotient.
After choosing, for each species permutation,
bijections between fibres of equal cardinality, the sequence splits,
giving a non-canonical semidirect-product description
\[
\AutP\bigl(\Lk_0(P)\bigr)\;\cong\;
\Bigl(\textstyle\prod_{(u,v)}\Symm(\Rx_{u,v})\Bigr)\rtimes\Symm(\Sp)_P ,
\]
in which $\Symm(\Sp)_P$ acts on the product of fibre symmetric groups by permuting
the fibres according to $\Rx_{u,v}\mapsto\Rx_{\pi u,\pi v}$.
\end{theorem}

\begin{proof}
By Theorem~\ref{thm:univ} a presentation-preserving endofunctor of $\Lk_0(P)$ is
determined by its restriction to the generators $\Sp$ and $\Rx$,
and is invertible if and only if that restriction is a bijection of each generating set respecting the typing.
On species this restriction is a permutation $\pi\in\Symm(\Sp)$;
on reactions it is a bijection $\beta\colon\Rx\to\Rx$ respecting typing up to $\pi$,
i.e. $s(\beta r)=\pi\,s(r)$ and $t(\beta r)=\pi\,t(r)$,
so $\beta$ restricts to a bijection $\Rx_{u,v}\to\Rx_{\pi u,\pi v}$ for every pair.
Such a $\beta$ exists \emph{if and only if} $|\Rx_{u,v}|=|\Rx_{\pi u,\pi v}|$ for all $(u,v)$,
which is exactly the condition $\pi\in\Symm(\Sp)_P$.
The assignment $(\pi,\beta)\mapsto$ (its functor) is thus a bijection onto $\AutP(\Lk_0(P))$.
Composing functors composes the pairs by $(\pi_1,\beta_1)(\pi_2,\beta_2)=(\pi_1\pi_2,\beta_1\beta_2)$,
the second component being the composite bijection
$\Rx_{u,v}\xrightarrow{\beta_2}\Rx_{\pi_2 u,\pi_2 v}\xrightarrow{\beta_1} \Rx_{\pi_1\pi_2 u,\pi_1\pi_2 v}$.
 
The pairs with $\pi=\id$ are exactly the $\beta$ preserving every fibre,
i.e. $K=\prod_{(u,v)}\Symm(\Rx_{u,v})$.
By the composition law the projection $(\pi,\beta)\mapsto\pi$ is a homomorphism $\AutP(\Lk_0(P))\to\Symm(\Sp)_P$;
it is surjective,
since the cardinality condition supplies a $\beta$ over each admissible $\pi$,
and its kernel is exactly $K$, which is therefore a normal subgroup.
For the splitting, fix once and for all a total order on each fibre,
and for $\pi\in\Symm(\Sp)_P$ let $\beta_\pi$ be the unique order-preserving bijection
$\Rx_{u,v}\to\Rx_{\pi u,\pi v}$ on each fibre 
(well-defined by $|\Rx_{u,v}|=|\Rx_{\pi u,\pi v}|$).
The composite of two order-preserving bijections is order-preserving,
so $\beta_{\pi_1\pi_2}=\beta_{\pi_1}\beta_{\pi_2}$
and $\pi\mapsto(\pi,\beta_\pi)$ is a homomorphism splitting the projection.
Write $H$ for its image.
Then $H\cap K=1$, since $\beta_\pi=\id$ forces $\pi=\id$; and
every automorphism factors as a $K$-part times an $H$-part,
because for any $(\pi,\beta)$ the bijection $\kappa:=\beta\circ\beta_\pi^{-1}$ fixes each species,
so $\kappa\in K$ and $(\pi,\beta)=\kappa\cdot(\pi,\beta_\pi)$.
Hence $\AutP(\Lk_0(P))=K\,H$ with $K$ normal and $H\cap K=1$,
an internal semidirect product.
Conjugation of $\kappa\in K$ by $(\pi,\beta_\pi)$ carries the fibre permutation on $\Rx_{u,v}$ to one on
$\Rx_{\pi u,\pi v}$,
which is the stated action of $\Symm(\Sp)_P$ on $\prod_{(u,v)}\Symm(\Rx_{u,v})$
by permuting fibres according to $\Rx_{u,v}\mapsto\Rx_{\pi u,\pi v}$.
\end{proof}
 
\begin{remark}[The splitting is non-canonical]
\label{rem:ses}
The invariant content of Theorem~\ref{thm:autoL0} is the short exact sequence
$$
1\to K\to\AutP(\Lk_0(P))\to\Symm(\Sp)_P\to 1,\qquad
K=\textstyle\prod_{(u,v)}\Symm(\Rx_{u,v}),
$$
together with its splitting;
the semidirect-product form depends on the choice of fibre orderings used to define the splitting,
and is canonical only up to that choice.
Only the normal factor $K$ and the quotient $\Symm(\Sp)_P$ are intrinsic, and only $K$ is used below.
\end{remark}

The normal factor $K$ is the source of all base-level indistinguishability:
whenever a fibre $\Rx_{u,v}$ has two or more reactions, $K$ swaps them, and so no
observable already present at $\Lk_0$ --- source, target, stoichiometry,
incidence --- tells parallel reactions apart. (A freshly chosen functor out of
$\Lk_0(P)$ can of course distinguish them; the point is about the data carried
\emph{by the level}, as made precise by the stabiliser below.) Each decoration
cuts $K$ down to the subgroup that preserves it, and the size of what survives
measures how much the decoration sees.

\begin{corollary}[A decoration breaks exactly the symmetries it detects]
\label{cor:stab}
Let $D\colon\Rx\to A$ be a decoration (e.g.\ $\FH$, valued in $\RR$) and write
$\Lk_0(P)_D:=(\Lk_0(P),D)$ for the decorated level. The base automorphisms
lifting to $\Lk_0(P)_D$ are exactly those in the stabiliser
\[
\Aut(\Lk_0(P),D):=\{\varphi\in\AutP(\Lk_0(P)) : D\circ\varphi=D\}.
\]
Within the normal factor $K$ this is
$K_D=\{\beta\in K : D(\beta r)=D(r)\ \text{for all }r\}$, the fibrewise
permutations preserving the level sets of $D$. Hence adjoining $D$ breaks
precisely the fibre symmetries that move a pair on which $D$ differs --- no more,
no fewer.
\end{corollary}

\begin{proof}
An automorphism of $\Lk_0(P)_D$ is
a base automorphism $\varphi$ that carries the decorated level to itself;
since the only added datum is $D$,
this holds exactly when $\varphi$ preserves it, i.e. $D\circ\varphi=D$,
which is the stabiliser condition.
As $D$ depends only on the reaction, $D\circ\varphi=D\circ\beta$,
so the condition involves only the fibre component $\beta$.
Restricting to $K$
(where $\pi=\id$ and $\beta$ acts within fibres),
$D\circ\beta=D$ says exactly that $\beta$ preserves the level sets of $D$.
The fibre permutations failing it are those moving a $D$-distinguished pair.
\end{proof}

Corollary~\ref{cor:stab} is the precise, and modest,
sense in which each level is \emph{forced}:
relative to the data carried so far,
the decoration $D$ breaks exactly the symmetries that conflate reactions $D$ separates,
and no others.
It gives a clean test for the \emph{fibre-forcing} part of non-collapse:
\emph{adjoining $D$ strictly refines the fibre-symmetry group iff $K_D\subsetneq K$},
i.e. iff some fibre contains two reactions on which $D$ differs.
(This is a statement about the normal factor $K$ only;
a decoration may also break a species symmetry in $\Symm(\Sp)_P$ while leaving $K$ untouched,
e.g.\ when all fibres are singletons so $K=1$,
so the test detects the parallel-reaction forcing, not every refinement.)
The forcing pairs of interest here live inside a single source--target fibre,
so $K_D\subsetneq K$ is the relevant criterion.
No claim is made that $\Lk_1$ is initial in some category of extensions;
the content is the stabiliser computation, which is true and sufficient.

\paragraph{The computation, run explicitly through $\Lk_3$.}
It is worth carrying the test out on a concrete network,
so that the abstract group of Theorem~\ref{thm:autoL0} is seen cutting down level by level.
Take the single fibre over the complex pair $(\ch{C}+\ch{O2},\,\ch{CO2})$ holding two parallel reactions,
\[
r_1,\,r_2\;\colon\;\ch{C}+\ch{O2}\longrightarrow\ch{CO2},
\]
the combustions of graphite and of diamond of Section~\ref{sec:L1}.
All other fibres are singletons,
so by Theorem~\ref{thm:autoL0} the entire base indistinguishability lives in this one fibre:
\[
K\;=\;\Symm(\Rx_{u,v})\;=\;\Symm(\{r_1,r_2\})\;\cong\;\ZZ/2,
\]
generated by the transposition $\tau$ swapping $r_1$ and $r_2$.
(The species factor $\Symm(\Sp)_P$ plays no role here,
as the forcing pair lies inside a fibre; it is $K$ that the decorations test.)
At $\Lk_0$ nothing distinguishes $r_1$ from $r_2$:
$\tau$ is a genuine symmetry of the bare network, $K=\ZZ/2$.

Now, adjoin the decorations of Section~\ref{sec:dict} in turn and apply Corollary~\ref{cor:stab},
$K_D=\{\beta\in K: D\circ\beta=D\}$, at each step.
\begin{itemize}[leftmargin=1.4em]
\item \emph{At $\Lk_1$.} The enthalpy functor takes the two combustion enthalpies,
which differ by the graphite--diamond transition enthalpy,
$\FH(r_1)\neq\FH(r_2)$. The swap $\tau$ does not preserve $\FH$,
so $\tau\notin K_{\FH}$ and
\[
K_{\FH}\;=\;\{\id\}\;\subsetneq\;K\;=\;\ZZ/2 .
\]
The inclusion is strict:
$\Lk_1$ is forced, and the broken generator $\tau$ is the lost symmetry the level records.
\item \emph{At $\Lk_2$.} Suppose instead a pair with $\FH(r_1)=\FH(r_2)$ but distinct entropies,
$\FS(r_1)\neq\FS(r_2)$ --- so that $\tau\in K_{\FH}$,
i.e. $\Lk_1$ does not separate them. 
Then $\tau$ fails to preserve the Gibbs datum, $\FG^T(r_1)\neq\FG^T(r_2)$, so
\[
K_{(\FH,\FS)}\;=\;\{\id\}\;\subsetneq\;K_{\FH}\;=\;\ZZ/2 .
\]
The entropy functor cuts the surviving group strictly further:
$\Lk_2$ is forced relative to $\Lk_1$.
\item \emph{At $\Lk_3$.} Suppose a pair sharing every thermodynamic datum,
$\FH(r_1)=\FH(r_2)$ and $\FS(r_1)=\FS(r_2)$,
but differing in rate: 
the uncatalysed and catalysed $\ch{CO2}$ hydrations of Section~\ref{sec:L3},
with $\FP(r_1)\neq\FP(r_2)$.
Now $\tau\in K_{(\FH,\FS)}$ but $\tau\notin K_{(\FH,\FS,\FP)}$, so
\[
K_{(\FH,\FS,\FP)}\;=\;\{\id\}\;\subsetneq\;K_{(\FH,\FS)}\;=\;\ZZ/2 .
\]
The kinetic functor is the first to separate the pair:
$\Lk_3$ is forced relative to $\Lk_2$.
\end{itemize}
In each case the surviving fibre group is read off Theorem~\ref{thm:autoL0} and
then cut by Corollary~\ref{cor:stab},
and the strict inclusion $K_D\subsetneq K_{D'}$ is exactly the statement that the level is not redundant.
The same computation runs at every decorated level because it is the same computation;
the surviving-symmetry groups for larger networks are tabulated in \cite[Table~3]{Han2026}.

\paragraph{The forcing class as a cokernel.}
The stabiliser computation has a dual reading that names the class the lower levels referred forward to.
For each step write
$\varphi_k\colon\AutP(\Lk_k(P))\to\AutP(\Lk_{k-1}(P))$
for the restriction that forgets the level-$k$ decoration,
giving the exact sequence of pointed sets
\[
1\;\to\;\ker\varphi_k\;\to\;\AutP(\Lk_k(P))\;
\xrightarrow{\ \varphi_k\ }\;\AutP(\Lk_{k-1}(P))\;\to\;\coker\varphi_k\;\to\;1,
\]
with $\coker\varphi_k:=\AutP(\Lk_{k-1}(P))/\operatorname{im}\varphi_k$ the set of left cosets,
based at the trivial coset \cite[\S2.5]{Han2026}.
Since $\operatorname{im}\varphi_k$ need not be normal,
$\coker\varphi_k$ is in general a pointed set rather than a group.
A non-trivial coset is represented by a symmetry of $\Lk_{k-1}(P)$ that does not lift to $\Lk_k(P)$,
that is,
by a base relabelling the decoration breaks: exactly the fibre permutation of Corollary~\ref{cor:stab}
that moves a pair on which the new datum differs.
The forcing test $K_D\subsetneq K$ is therefore the statement that $\coker\varphi_k$ is non-trivial,
and the lost symmetry recorded at each of $\Lk_1,\Lk_2,\Lk_3$ is the corresponding cokernel class.
What does \emph{not} reduce to this pattern is the step to $\Lk_4$,
where the distinction being made is not purely numerical.
\subsubsection{$\Lk_4$: mechanism as a base change}
\label{sec:L4}

\paragraph{What $\Lk_3$ cannot see: mechanism.}
Two reactions can have the same net stoichiometry,
the same equilibrium datum,
and the same mass-action profile in the chosen effective $\Lk_3$ description,
and still differ in \emph{mechanism}:
in which bonds break and form,
and whether the transformation is \emph{concerted} (a single elementary step)
or \emph{stepwise} (passing through a discrete intermediate).
Such a pair has the same $\Lk_3$-profile $(s,t,\FH,\FS,\FP)$,
even though the two remain distinct reaction labels
--- distinct morphisms of the syntax ---
at every level $\Lk_0,\dots,\Lk_3$.
The data \emph{distinguishing} them is not numerical:
no observable factoring through the lower profile $(s,t,\FH,\FS,\FP)$ separates a concerted from a stepwise route,
since by hypothesis they agree on all of it. 
Reaction mechanism is recoverable only by recording
\emph{how} the transformation factors through molecular structure ---
which bonds move, through what intermediate ---
and this is the one step of the tower so far that changes the underlying category rather than decorating it.

One overclaim must be pre-empted.
One \emph{can} of course label reactions by mechanism by hand
--- a function $M\colon\Rx\to\{\text{concerted},
\text{stepwise}\}$ extends,
by the universal property,
to a functor out of $\Lk_0(P)$.
The content of $\Lk_4$ is not that mechanism is undefinable below it,
but that mechanism is not recoverable \emph{from the lower observables}:
it is not a function of $(\FH,\FS,\FP)$, hence genuinely new data,
and $\Lk_4$ supplies its canonical structural source
--- factorisation through molecular graphs ---
rather than an ad hoc label.

\paragraph{The ambient category of molecular graphs.}
Following the graph-rewriting tradition \cite{Ehrig1973,LackSobocinski},
species-as-tokens are replaced by species-as-graphs.
A \emph{labelled molecular graph} is a finite graph whose vertices carry element type (atomic number),
formal charge, lone-pair and radical counts and whose edges carry bond order;
let $\LGraph$ be the (adhesive) category of these with label-preserving monomorphisms.
To read off the Petri-net boundary,
a graph is equipped with a \emph{boundary species assignment} on its connected components,
following \cite[Def.~6.13]{Han2026}.

\begin{definition}[Grounded molecular graphs]
\label{def:LGraphP}
$\LGraphP$ has as objects pairs $(G,\mathrm{spec})$,
where $G\in\LGraph$ and $\mathrm{spec}\colon\pi_0(G)\to\Sp$ assigns
a $P$-species to each connected component (each individual molecule) of $G$.
Its morphisms are the underlying label-preserving monomorphisms of $\LGraph$;
no compatibility with $\mathrm{spec}$ is imposed,
since $\mathrm{spec}$ is boundary data rather than rewrite data.
The rewriting itself takes place in $\LGraph$,
taken to be adhesive so that double-pushout (DPO) rewriting is well behaved \cite{LackSobocinski};
the assignment $\mathrm{spec}$ is used only to record which species a molecule represents,
i.e. to read the Petri-net boundary
(so reactions merging or splitting components are unobstructed).
The molecular label algebra is detailed in \cite[\S6.2]{Han2026}.
\end{definition}

A \emph{DPO rule} is a span $L\hookleftarrow K\hookrightarrow R$ in $\LGraph$
($L$ the reactant pattern, $R$ the product, $K$ the preserved context);
it is \emph{chemically admissible} when it satisfies the local label-balance conditions of the chosen molecular-graph formalism
--- atom identities preserved, and changes in bond order, charge, lone-pair
and radical count obeying the allowed elementary rewrite schema \cite[\S6.3]{Han2026}.
The full label algebra is not reproduced here;
for this paper it suffices that admissibility is a decidable local condition.

The status of $\Lk_4$ is fixed as \emph{syntax},
with each generator tied to the network.
A \emph{chemical DPO rule} for $P$ is a fully instantiated,
species-level rule symbol:
a span $p$ in $\LGraph$ whose reactant and product graphs $L$, $R$ realise specific complexes of $P$
(not merely an ungrounded local rewrite schema),
together with a \emph{generator label} $\gamma(p)\in\Rx$:
the reaction of $P$ whose net transformation the rule realises,
so that the boundary species of $L$ and $R$ are the complexes $s(\gamma(p))$ and $t(\gamma(p))$.
The label is part of the rule data, exactly as in \cite[Def.~6.14]{Han2026};
it is what lets a mechanistic rule be projected to a specific Petri-net generator rather than to an ambiguous source--target pair,
and full instantiation is what makes $\gamma(p)$, and hence the forgetful image below, well defined.

\begin{definition}[Mechanistic level]
\label{def:L4}
$\Lk_4(P)$ is the free strict symmetric monoidal category on the signature of chemical DPO rules for $P$:
objects are finite disjoint unions of grounded molecular graphs in $\LGraphP$,
generating morphisms are the admissible labelled rules $p$
(each typed $L\to R$ by its reactant and product graphs and carrying its label $\gamma(p)$),
and morphisms are formal composites and tensors of generating rules and structural symmetries,
modulo the strict symmetric monoidal congruence.
DPO derivability in $\LGraph$ interprets each rule symbol.
\end{definition}

\begin{proposition}[Universal property of $\Lk_4(P)$]
\label{prop:L4free}
$\Lk_4(P)$ is the free strict SMC on the signature of chemical DPO rules for $P$:
for every strict SMC $C$,
restricting a strict symmetric monoidal functor $F\colon\Lk_4(P)\to C$ to its generators
--- an object of $C$ for each grounded molecular graph and,
for each rule $p$ typed $L\to R$,
a morphism between the images of $L$ and $R$ ---
is a bijection onto all such generator assignments, and is natural in $C$.
\end{proposition}
 
\begin{proof}
For uniqueness, every morphism of $\Lk_4(P)$ is, by Definition~\ref{def:L4},
a formal composite and tensor of generating rules and structural symmetries,
taken modulo the strict symmetric monoidal congruence.
A strict symmetric monoidal functor preserves composition, tensor, and the symmetries,
so $F$ is determined on every morphism once its values on the generators are fixed:
the images of the symmetries are forced by strictness,
and the image of a composite or tensor is the corresponding composite or tensor of images.
Two functors agreeing on the generators therefore agree everywhere,
giving injectivity of the restriction.
 
For existence,
given a generator assignment of the stated form,
define $F$ on generators by it and extend to composites and tensors by the functor and
monoidality equations.
This is well defined on congruence classes because $C$,
being a strict SMC,
satisfies exactly the relations generating the congruence:
the interchange law,
the naturality and coherence of the symmetry,
and strict associativity and unit.
Hence congruent formal expressions receive equal images,
and $F$ descends to $\Lk_4(P)$.
By construction $F$ is strict symmetric monoidal and restricts to the given assignment,
giving surjectivity.
Naturality in $C$ is immediate,
since post-composing $F$ with a strict symmetric monoidal $\phi$ post-composes the generator values with $\phi$.
\end{proof}

\paragraph{The forgetful functor $U_4$.}
The mechanistic level is tied back to the rest of the tower by a forgetful functor,
exactly as the decorated levels are
--- the difference being that $U_4$ goes between two \emph{different} free categories rather than dropping a functor.

\begin{proposition}[The forgetful functor $U_4$]
\label{prop:U4}
Stripping all bond-level structure defines a strict symmetric monoidal functor
$U_4\colon\Lk_4(P)\to\Lk_0(P)$ that on objects
sends a grounded graph
$(G,\mathrm{spec})$ to the multiset of species of its components counted by $\mathrm{spec}$,
and on a generating rule $p$ returns the labelled generator $\gamma(p)\in\Rx$,
regarded as a morphism of $\Lk_0(P)$.
The lower decorations lift along it, with $\FH\circ U_4$, $\FS\circ U_4$, and $\FP\circ U_4$
the energetic, equilibrium, and kinetic data of $\Lk_4(P)$.
Assembling these gives the \emph{intermediate forgetful functor}
\[
U_4^{\Lk_3}\colon\Lk_4(P)\to\Lk_3(P),\qquad
U_4^{\Lk_3}(d)=\bigl(U_4(d),\,\FH U_4(d),\,\FS U_4(d),\,\FP U_4(d)\bigr),
\]
whose underlying strict symmetric monoidal functor is $U_4$ and which satisfies
$U_4=U_1\circ U_2\circ U_3\circ U_4^{\Lk_3}$.
Thus $\Lk_4(P)$ retains all lower information and adds only the bond-level structure carried by its morphisms
\cite[\S6.6, Def.~6.30, Prop.~6.31]{Han2026}.
\end{proposition}
 
\begin{proof}
The object and generator assignments above form a generator assignment in the sense of Proposition~\ref{prop:L4free},
with target $\Lk_0(P)$: each grounded graph is sent to a species multiset,
an object of $\Lk_0(P)$, and each generating rule $p$ to the morphism $\gamma(p)$.
This is type-correct
because grounding fixes the boundary species of the reactant and product graphs $L,R$ of $p$
to be the complexes $s(\gamma(p))$ and $t(\gamma(p))$,
so $\gamma(p)$ is indeed a morphism between the images of $L$ and $R$.
The label is needed here:
$\Lk_0(P)$ is not thin,
so a rule cannot be sent to ``the morphism with matching source and target'' (there may be several),
and $\gamma(p)$ resolves the choice \cite[Def.~6.27, Rem.~6.28]{Han2026}.
By the universal property of Proposition~\ref{prop:L4free}
this assignment extends uniquely to a strict symmetric monoidal functor $U_4\colon\Lk_4(P)\to\Lk_0(P)$,
which is therefore well defined.
 
Composing $U_4$ with the decoration functors of the lower tower gives functors
$\FH\circ U_4$, $\FS\circ U_4$, $\FP\circ U_4$ out of $\Lk_4(P)$,
the lifted energetic, equilibrium, and kinetic data.
Pairing $U_4$ with these defines $U_4^{\Lk_3}$ as displayed;
its underlying strict symmetric monoidal functor is $U_4$ by construction.
Finally, each $U_k$ drops the topmost decoration at its level,
so $U_3$ applied to $U_4^{\Lk_3}$ discards $\FP U_4$ and returns the $\Lk_2$-datum,
and chaining $U_1\circ U_2\circ U_3$ removes the remaining decorations, leaving $U_4$;
this is the factorisation $U_4=U_1\circ U_2\circ U_3\circ U_4^{\Lk_3}$ \cite[Def.~6.30, Prop.~6.31]{Han2026}.
\end{proof}

This places $\Lk_4$ alongside the rest of the tower as a refinement of the level below (Section~\ref{sec:forcing}): 
the question is what the effective $\Lk_3$ description cannot record,
and the answer is carried by the morphisms of $\Lk_4(P)$ rather than by a surviving-symmetry computation.

\paragraph{What $\Lk_4$ adds to $\Lk_3$: the structural separation.}
The same chemistry admits two Petri-net descriptions at different resolution.
The \emph{elementary} network $P_{\mathrm{elem}}$ lists every elementary step and
treats each short-lived intermediate as a species of its own.
The \emph{effective} network $P_{\mathrm{eff}}$ is
obtained from it by a quasi-steady-state (QSS) reduction that eliminates those intermediates,
so that a sub-mechanism collapses to a single reaction carrying an aggregated rate constant \cite[\S5.8]{Han2026}; 
this is the network the rest of the development works with,
the one whose reactions are the effective channels of $\Lk_0$--$\Lk_3$.
The reduction is a functor
$$
Q\colon \Lk_3(P_{\mathrm{elem}})\longrightarrow \Lk_3(P_{\mathrm{eff}})
$$
between the kinetic levels of the two descriptions.
 
The two mechanistic routes both live in $\Lk_4(P_{\mathrm{elem}})$,
where intermediates are present:
a concerted route is a \emph{single} generating DPO rule,
a stepwise route through the intermediate a genuine \emph{two-step} composite $p_{\mathrm{elim}}\circ p_{\mathrm{add}}$.
In the elementary description these are already different at $\Lk_3$:
the stepwise route has the intermediate as a species and a two-step kinetics.
They become indistinguishable only after $Q$.
Where the QSS reduction is valid,
the stepwise route's image under $Q$ is a single reaction sharing every $\Lk_3$ datum with the concerted route,
and no bulk-kinetic observable distinguishes the two;
where it is not valid, $Q$ does not apply,
the stepwise route keeps its two-step kinetics,
and $\Lk_3$ separates the routes on its own.
The conflation is therefore contingent on the coarse-graining,
not categorical \cite[\S6.5]{Han2026};
what $\Lk_4$ contributes is a structural record of the mechanism that survives the reduction,
available whether or not the effective kinetics happens to distinguish the routes.
The distinction it supplies lives upstream of $Q$,
in the derivation length in $\Lk_4(P_{\mathrm{elem}})$.
 
\begin{proposition}[Mechanism refines the effective profile]
\label{prop:Dmech}
Mechanism is not a function of the effective profile $(s,t,\FH,\FS,\FP)$: a single
effective channel can be realised by DPO derivations of different length, a
distinction no decoration factoring through that profile can record.

Precisely, let $r\in\Rx(P_{\mathrm{eff}})$ be an effective reaction with profile
$(s,t,\FH,\FS,\FP)$, realised by two chemically distinct routes --- a concerted
route, given by a single chemical DPO rule $p_{\mathrm c}$ in
$\Lk_4(P_{\mathrm{elem}})$, and a stepwise route, given by a proper composite
$p_{\mathrm{elim}}\circ p_{\mathrm{add}}$ through an intermediate species --- both
with effective image $r$ under $Q\circ U_4^{\Lk_3}$. Then $p_{\mathrm c}$ and
$p_{\mathrm{elim}}\circ p_{\mathrm{add}}$ are distinct morphisms of
$\Lk_4(P_{\mathrm{elem}})$, one a generating rule and the other a proper composite,
with the same image under $Q\circ U_4^{\Lk_3}$. Thus $\Lk_4$ does not merely
decorate the effective channel; it replaces effective reactions by mechanistic
factorisations.
\end{proposition}
 
\begin{proof}
The two routes are distinct already in $\Lk_4(P_{\mathrm{elem}})$.
The concerted route is the single generator $p_{\mathrm c}$;
the stepwise route is the composite $p_{\mathrm{elim}}\circ p_{\mathrm{add}}$,
which by freeness (Proposition~\ref{prop:L4free}) is not a generating rule,
since the symmetric monoidal congruence relates only formal expressions of equal generator-length
and the symmetries it introduces contribute no generators;
a two-generator composite is therefore never identified with a single generator.
Their effective images agree.
Under $U_4^{\Lk_3}\colon\Lk_4(P_{\mathrm{elem}})\to\Lk_3(P_{\mathrm{elem}})$ the
generator maps to a one-step elementary reaction,
on which the coarse-graining $Q$ acts trivially, returning $r$;
the composite maps to the two-step path through the intermediate,
which $Q$ collapses to the same $r$ with its aggregated profile \cite[\S5.8]{Han2026}.
Hence both have image $r$ under $Q\circ U_4^{\Lk_3}$,
while remaining distinct morphisms of $\Lk_4(P_{\mathrm{elem}})$ by the length argument above.
Any decoration factoring through the effective profile $(s,t,\FH,\FS,\FP)$ identifies the two,
since their profiles agree by hypothesis;
what separates them is the derivation length, which that profile does not see.
The chemical instance discharging this proposition is treated in the application section
\cite[Prop.~6.40]{Han2026}.
\end{proof}

Like the decorated levels of Section~\ref{sec:forcing},
this records a distinction the lower level cannot express; unlike them,
it does so by changing the category rather than adding a functor.
The extension is structural, not numerical:
$U_4$ replaces the objects and morphisms of the level below,
so the new datum is carried by the derivations themselves,
and the separation is detected by inequality of those derivations rather than by a symmetry that fails to survive. 
This sharpens Remark~\ref{rem:generators},
where additive generator bookkeeping was seen to forget sequential structure: the structure it forgets is,
in this instance,
exactly the two-step factorisation $p_{\mathrm{elim}}\circ p_{\mathrm{add}}$ that $\Lk_4$ records and the coarse-graining $Q$ discards.
 
The forcing form of the lower levels does reappear,
but only contingently.
On the effective network, where $Q$ has already identified the two routes, 
their swap is an automorphism of $\Lk_3(P_{\mathrm{eff}})$ that does not lift,
a non-trivial $\coker(\varphi_4)$ in the sense of Section~\ref{sec:forcing}.
That reading is available only after $Q$ and only in its regime,
whereas the structural statement needs neither;
the structural statement is therefore the primary one.
A concrete instance is worked in the application section,
where it also shows why the distinction matters for network-level questions.

\paragraph{A half-step: stereochemistry.}
Between mechanism and geometry sits one further refinement,
needed in the sequel but slight enough to record here rather than in a section of its own.
A DPO rule fixes which bonds break and form but not the three-dimensional arrangement of the atoms around a reacting centre,
so two reactions realising the same rule can still differ in stereochemical outcome
--- retention versus inversion at a centre,
formation of one enantiomer versus its mirror image.
The distinction is invisible at $\Lk_4$,
where the rule is the same;
it is the same kind of gap the tower has met before,
and it is closed by adding the missing datum.
 
The added datum is a chirality label.
Each molecular graph is augmented with a sign at every stereocentre,
a map $\sigma\colon\mathrm{Stereo}(G)\to\{\pm1\}$ recording its handedness,
and the chirality-symmetry group $\Gstar=\Aut(G)\ltimes\mathbb{Z}_2^{k}$ acts on the labelled graphs,
its $\mathbb{Z}_2^{k}$ factor flipping the signs at the $k$ stereocentres \cite[\S7]{Han2026}.
The level $\Lk_{4.5}(P)$ is built from $\Gstar$-\emph{equivariant} DPO rules:
rules required to act compatibly with this action,
so that a rule now carries how its reacting centres transform under chirality operations.
Two reactions that share a bare $\Lk_4$ rule but produce opposite handedness
--- a mechanism and its mirror image ---
differ in their $\sigma$-signs and so become distinct equivariant rules at $\Lk_{4.5}$.
 
The forgetful functor $U_{4.5}\colon\Lk_{4.5}(P)\to\Lk_4(P)$ drops the chirality label and the equivariance.
Its forcing is again a non-trivial $\coker(\varphi_{4.5})$:
the swap of two enantiomeric mechanisms is an automorphism of $\Lk_4(P)$, where the rule is the same,
that does not lift to $\Lk_{4.5}(P)$,
where the opposite $\sigma$-signs make the two no longer the same equivariant rule.
This is an enrichment of the morphisms by a symmetry,
a third type of extension after decoration and base change \cite[Rem.~7.24]{Han2026},
and it is all of $\Lk_{4.5}$ the present development uses:
the geometric and electronic-structure levels of Sections~\ref{sec:L5}--\ref{sec:L6} are built over $\Lk_{4.5}(P)$, 
and their forcing pairs are reactions already identical at $\Lk_{4.5}$.
\subsubsection{$\Lk_5$: geometry, when the datum becomes a function on a manifold}
\label{sec:L5}

The levels to this point have been combinatorial.
Every datum is a finite assignment to generators:
a real number, a Markov generator, a labelled DPO rule.
The next two levels leave that world,
and they record the physical structure internal to a single molecular species rather than the composition of a network,
and the datum that each adjoins is no longer finite.

At $\Lk_5$ it is a function on a continuous configuration space,
and at $\Lk_6$ (Section~\ref{sec:L6}) a topological invariant of a bundle over that space.
The forcing apparatus above still applies, but in a new mode.
The decorated levels were forced by \emph{separation}:
a decoration broke a base symmetry that conflated two reactions.
The geometric and electronic levels are forced by \emph{object refinement}:
a single lower object splits into distinct upper ones,
and the new datum is what the lower level could record as a free number but not \emph{derive}.
Their detailed constructions are given in \cite{Han2026};
here, only the data needed for the present argument are recorded,
since this is where the genuinely physical content of chemistry first enters the formal description.

\paragraph{What $\Lk_{4.5}$ cannot do: explain an isotope rate difference.}
Consider the $\mathrm{S_N2}$ substitution run with two isotopologues,
\[
r_{\mathrm H}\colon \ch{CH3Br} + \ch{OH-}\to\ch{CH3OH}+\ch{Br-},\qquad
r_{\mathrm D}\colon \ch{CD3Br} + \ch{OH-}\to\ch{CD3OH}+\ch{Br-}.
\]
The two reactions run at measurably different rates \cite[\S7.7]{Han2026}.
The respective rate constants $k_{\mathrm H}$ and $k_{\mathrm D}$ are
the $\Lk_3$ data $\FP(r_{\mathrm H})$ and $\FP(r_{\mathrm D})$ (respectively),
and they are not equal.
The lower tower records the two values but does not explain them,
and nothing below $\Lk_5$ says why they differ.

Up to $\Lk_{4.5}$, the two substrates are the same object.
The vertex label in the grounded category $\mathbf{LGraph}_P$ records the \emph{element},
i.e. the atomic number.
Hydrogen and deuterium are the same element.
The replacement of every H by D therefore leaves the labelled graph the same,
the equivariant DPO rule,
and the chirality lift unchanged,
and neither carbon is a stereocentre.
The swap $r_{\mathrm H}\leftrightarrow r_{\mathrm D}$ is an automorphism of $\Lk_{4.5}(P)$.

One might object that the mass could simply be added as a vertex label,
and the difference declared a new stoichiometric datum.
A discrete tag of this kind is, in fact, available:
the label algebra of \cite[Def.~6.2]{Han2026} carries an $\mathrm{iso}$ component,
and the mechanistic level uses it to track atoms,
for instance to follow an $^{18}$O label through a substitution \cite[Rem.~6.3]{Han2026}.
But a tag cannot do the work required here 
for a reason that is physical and notational.
The chemistry below the metric is governed by electronic structure.
The electronic Hamiltonian, and thus the potential energy surface it determines,
depends on the nuclear \emph{charge} and position, but not on the nuclear \emph{mass} \cite{BornOppenheimer1927}.
Therefore, two isotopologues share one and the same potential energy surface;
they differ only in nuclear kinetic energy \cite[Ch.~2]{Wolfsberg2009}.
Mass enters chemistry through that kinetic energy alone,
and it enters as a continuous quantity, not as a discrete type.
A label can record \emph{that} two isotopologues differ,
but it cannot supply the structure from which the difference follows.
That structure is the mass-weighted metric on the space of nuclear positions,
and it is the datum $\Lk_5$ adjoins \cite[Rem.~8.2]{Han2026}.

This also shows why $\Lk_5$ is the \emph{minimal} extension of $\Lk_{4.5}$ that accounts for the difference.
Mass cannot be carried by any datum the lower levels admit.
A decoration on $\Lk_0(P)$ is a number attached to a reaction; it can store the rate, but storing is not deriving, and a stored number explains nothing.
A vertex label is a discrete type; it cannot carry a continuous mass, and if it is added it sits inert at the graph level, where it changes neither the rule nor any lower datum.
Mass acquires force only when it weights a kinetic energy, and a kinetic energy is a quadratic form on velocities, that is a metric.
A metric in turn requires a space of nuclear geometries to live on; indeed the mass-weighted metric descends to that space only when mass is constant on each graph-automorphism orbit \cite[Rem.~8.2]{Han2026}, so mass and geometry must be introduced together.
The configuration space with its mass-weighted metric is thus the least structure that lets an isotope difference follow from anything, and that is exactly $\Lk_5$.

The forcing here is therefore of a new kind.
At the decorated levels a decoration broke a symmetry that conflated two reactions.
At $\Lk_5$ the rate is already recorded below, as a free number; what is new is that the geometry \emph{constrains} that number rather than separating a pair.
The mechanism is the change in the carbon--hydrogen out-of-plane bending vibrations between the reactant and the transition state.
That change shifts the zero-point energy by an amount that depends on nuclear mass, and the shift differs between H and D \cite[\S8.1]{Wolfsberg2009}.
Computing it requires two data that exist nowhere below $\Lk_5$: the bending force constants, which are entries of a mass-weighted Hessian, and the transition-state geometry at which they are evaluated \cite[\S8.1]{Han2026}.
Both presuppose three-dimensional nuclear positions, a curvature of the energy landscape, and a nuclear mass.

\paragraph{The geometric decoration.} The level works throughout within the
Born--Oppenheimer approximation \cite[\S8.2]{Han2026}, in which electronic and nuclear motion decouple at leading order in $(m_e/M)^{1/2}$ and the electrons relax to their ground state at each fixed nuclear geometry.
It adjoins to each molecular species a \emph{Morse triple} $(\Ce(G),V,g)$:
\begin{enumerate}[label=(\alph*),leftmargin=2.2em]
\item the \emph{configuration orbifold} $\Ce(G):=\RR^{3n}/(\SE(3)\times\Autmu(G))$,
the space of nuclear geometries modulo rigid motions and mass-preserving graph automorphisms
(an orbifold, not a manifold: symmetric configurations have non-trivial stabilisers) \cite[Def.~8.1]{Han2026};
\item the \emph{potential energy surface} $V\colon\Ce(G)\to\RR$, 
the Born--Oppenheimer ground-state electronic energy plus nuclear repulsion as a function of geometry,
whose minima are conformers, index-1 saddles transition states,
and steepest-descent paths the intrinsic reaction coordinates \cite[Def.~8.14]{Han2026};
\item the \emph{mass-weighted Riemannian metric} $g$ induced by the nuclear kinetic energy:
in Cartesian nuclear coordinates $g_{i\alpha,j\beta}=m_i\delta_{ij}\delta_{\alpha\beta}$,
and after passing to mass-weighted coordinates it becomes Euclidean,
so that normal modes, zero-point energies, and the reaction-coordinate geometry acquire their meaning.
\end{enumerate}
These are packaged as a functor $\FV\colon\Lk_{4.5}(P)\to\OrbMorse$ into the category of Morse triples.
Its morphisms are gradient-flow cobordisms, the intrinsic reaction coordinates through an index-1 saddle, composed by concatenation of paths,
with identities the constant flows;
associativity and unitality are part of the construction \cite[Def.~8.24, Def.~8.29]{Han2026}.

\begin{definition}[Geometric level {\cite[Def.~8.35]{Han2026}}]
\label{def:L5}
$\Lk_5(P):=(\Lk_{4.5}(P),\FV)$, with $\FV\colon\Lk_{4.5}(P)\to\OrbMorse$ the geometric functor above and $U_5$ forgetting $\FV$.
\end{definition}

Two structural facts mark this as a new \emph{kind} of extension.
First,
and most relevant categorically, $\FV$ is the first \emph{lax} monoidal functor in the tower:
the lower decorators $\FH,\FS,\FP$ were strict, but the potential surface of two interacting fragments is not the sum of their separate surfaces,
so $\FV(G\otimes H)$ and $\FV(G)\otimes\FV(H)$ are related only by a laxator $\phi_{G,H}\colon\FV(G)\otimes\FV(H)\to\FV(G\otimes H)$,
not an isomorphism, subject to the usual associativity and unit coherence \cite[Def.~8.29, Rem.~8.31]{Han2026}.
Second,
the datum is infinite-dimensional:
specifying a reaction at $\Lk_5$ means giving the shape of $V$ along the reaction path,
not one or two real numbers per generator \cite[\S8.1]{Han2026}.
Where $\Lk_0$--$\Lk_3$ adjoined a strict decoration, $\Lk_3\to\Lk_4$ changed the morphisms,
and $\Lk_4\to\Lk_{4.5}$ enriched by a symmetry, $\Lk_{4.5}\to\Lk_5$ is a \emph{geometric decoration}:
the fourth type of extension the tower exhibits \cite[Rem.~7.24]{Han2026}.

\paragraph{Coherence: kinetics as a derived quantity.}
The geometric level does not merely sit above the kinetic one; it constrains it.
Transition-state theory computes a rate constant from the Morse triple:
from $\FV(r)$ one extracts the barrier height $V^{\ddagger}(r)$ and the harmonic frequencies at the reactant minimum and the index-1 saddle,
and assembles them into the transition-state rate $k_{\mathrm{TST}}\bigl(\FV(r),T\bigr)$,
a partition-function ratio times an Arrhenius factor \cite[Def.~8.38]{Han2026}.
With a transmission coefficient $\kappa(r,T)$ for recrossing and tunnelling,
itself a functional of $(\FV,g)$, the coherence condition is
\[
\FP(r)\;=\;\kappa(r,T)\,k_{\mathrm{TST}}\bigl(\FV(r),T\bigr).
\]
This is not asserted to hold exactly:
it holds up to anharmonic and deep-tunnelling corrections for elementary reactions in the Born--Oppenheimer regime,
whereas the naive form $\kappa\equiv1$ with harmonic,
ideal-gas $k_{\mathrm{TST}}$ essentially never does \cite[Def.~8.38, Prop.~8.40]{Han2026}.
The $\Lk_1$ enthalpy is recovered likewise from the potential surface and Hessian \cite[Prop.~8.43]{Han2026}.
Thus $\FV$ does not replace the lower functors but constrains their values:
$\FP$ and $\FH$, free numerical decorations below,
become quantities the geometry determines up to these stated corrections.

\paragraph{The forcing.} The construction is carried out in \cite{Han2026};
the present paper records its statement.
As in Section~\ref{sec:forcing},
the forgetful functor $U_5\colon\Lk_5(P)\to\Lk_{4.5}(P)$ induces by restriction
$\varphi_5\colon\Aut(\Lk_5(P))\to\Aut(\Lk_{4.5}(P))$,
and the exact sequence
\[
1\to\ker\varphi_5\to\Aut(\Lk_5(P))\xrightarrow{\ \varphi_5\ }\Aut(\Lk_{4.5}(P))
\to\coker(\varphi_5)\to 1
\]
has non-trivial cokernel \cite[Prop.~8.36]{Han2026}.
The non-triviality is an object refinement rather than a separation of a conflated pair.
At leading Born--Oppenheimer order the two isotopologues share the same potential surface $V$ and differ only in the mass-weighted metric, $g_{\mathrm H}\neq g_{\mathrm D}$, for the reason given above.
They are the same object of $\Lk_{4.5}(P)$,
the graph label being isotope-blind,
so the swap $r_{\mathrm H}\leftrightarrow r_{\mathrm D}$ is an automorphism there;
but their Morse triples $(\Ce(G),V,g_{\mathrm H})$ and $(\Ce(G),V,g_{\mathrm D})$ are distinct objects of $\OrbMorse$,
so the single $\Lk_{4.5}$-object splits into two at $\Lk_5$ and the swap is no longer an automorphism.
Its class is the recorded non-trivial element of $\coker(\varphi_5)$ \cite[Prop.~8.36, Prop.~8.40]{Han2026}.
This is the explanatory forcing of the introduction made precise.
The mass datum that splits the two metrics is exactly what the lower tower could record but not derive: the rate difference between the isotopologues.
\subsubsection{$\Lk_6$: electronic structure, when the datum becomes a topological invariant}
\label{sec:L6}

\paragraph{What $\Lk_5$ cannot see, and why.}
The whole of $\Lk_5$ rests on one assumption:
that the electronic ground state is separated from the first excited state by a spectral gap throughout the accessible region \cite[\S8.3]{Han2026}.
The gap is what makes the chemistry of $\Lk_5$ well defined.
Where the gap is positive,
perturbation theory gives a smooth choice of ground-state electronic wavefunction at each nuclear geometry \cite{Kato1966},
so the ground-state surface $V$ is a smooth function on $\Ce(G)$ and the Morse triple $(\Ce(G),V,g)$ carries everything.

However, it fails at a \emph{conical intersection} (CI),
where the two surfaces touch and the gap closes \cite{LonguetHiggins}.
The failure is not a loss of precision but a change in topology,
and it is invisible to $V$.
Away from the intersection, the eigenvalue $V$ stays smooth,
so its values near a CI look like values anywhere else;
what a CI changes is the global twisting of the electronic ground state over the surrounding configurations,
a fact no value of $V$ records \cite[\S8.7]{Han2026}.
This is the gap that $\Lk_6$ fills.

Two electronic states can sit over the \emph{same} Morse triple $(\Ce(G),V,g)$ and agree on every $\Lk_5$ datum,
the same minima, the same barrier, the same transition-state-theory rate,
yet differ in whether a CI is enclosed.
The datum that separates them is not a number or a function but a topological invariant of the electronic states over the punctured configuration space.

\paragraph{The witnesses.}
Three molecular systems make the pattern concrete.
They are witnesses of the phenomenon,
not a literal pair of reactions swapped by an automorphism of one $\Lk_5(P)$.
The canonical one is the sodium trimer $\ch{Na3}$.
Its Jahn--Teller ground state has a symmetry-required conical intersection at the equilateral $D_{3h}$ geometry, the Mexican-hat intersection of the $E\otimes e$ problem \cite{mayer1996rovibronic,meiswinkel1991pseudo}.
A loop in the pseudorotation coordinate encircles the seam once.
The surface and metric are common to two descriptions of the nuclear motion, one with a single-valued wavefunction and one carrying the sign change of a loop around the seam.
The first predicts integer pseudorotation quanta, the second half-integer quanta, and only the second matches the observed spectrum \cite[\S9.1]{Han2026}.
The hydrogen-exchange reaction $\ch{H + H2}$ is a second.
Its surface is smooth and its computed rate agrees with experiment, so at $\Lk_5$ it looks complete; yet a CI sits at the $D_{3h}$ geometry of the $\ch{H3}$ system, off the reaction path but inside the region the nuclear wavefunction explores, and the sign change of a loop around it appears as an interference shift in the measured cross-section of the isotopic variant $\ch{H + HD}$ \cite{Yuan2020}.
The $\mathrm{S_N2}$ substitution of the previous level is the negative case: it is CI-free, and $\Lk_5$ captures it completely.
In each system the discriminator is the parity of seam encirclement, an obstruction in $\ZZ/2$ that no correction to $V$ can supply.

\paragraph{Why the new datum is a bundle.}
The right object to carry this content is not a function but a bundle, and the reason is classical.
Berry showed that a quantum state carried adiabatically around a loop in parameter space returns with a geometric phase \cite{Berry}; Simon identified that phase as the holonomy of a connection on a line bundle of eigenstates, so that the geometric phase is a statement about bundle topology rather than about any one state \cite{Simon1983}.
In a molecule the parameter space is the nuclear configuration space, and the connection is the one Mead and Truhlar introduced as a vector potential in the nuclear Schr\"odinger equation, the molecular Aharonov--Bohm effect \cite{MeadTruhlar1979}.
The rigorous adiabatic theory that controls the errors of this decomposition is due to Panati, Spohn and Teufel \cite{PanatiSpohnTeufel2003}.
The bundle is thus not a new postulate but the object this established machinery already lives on.
$\Lk_5$ keeps only its scalar shadow, the eigenvalue $V$; $\Lk_6$ records the bundle itself.

\paragraph{The electronic decoration.}
Over an adiabatic region $\Omega\subseteq\Ce(G)$ on which a rank-$N$ cluster of electronic states is spectrally isolated, $\Lk_6$ adjoins the isolated adiabatic Hilbert bundle $\Hel$ with its Berry connection, the codimension-2 CI seam $\Xseam\subset\Omega$, and the \emph{real Berry-sign class} \cite[Def.~9.1, Def.~9.16]{Han2026}
\[
\eta_B \;=\; w_1(\Lzero^{\RR}) \;\in\; H^1(\Omega\setminus\Xseam;\,\ZZ/2),
\]
the first Stiefel--Whitney class of the real ground-state eigenline, which evaluates non-trivially on a meridian loop linking the seam \cite{ahn2019stiefel}.
The codimension fixes which invariant is informative.
For a real-symmetric Hamiltonian, the regime of thermal molecular chemistry, a degeneracy imposes two real conditions, so the seam has codimension~2, a loop around it is a circle, and the invariant is the $\ZZ/2$ class $w_1$ \cite{LonguetHiggins}.
When time-reversal symmetry is broken the seam has codimension~3, a small sphere replaces the loop, and the Chern class $c_1(\Lzero)\in H^2$ takes its place, the Berry-monopole regime \cite[\S8.7]{Han2026}.
These data form a category $\HilbBundR$ of electronic bundles with sign class, whose morphisms are geometric channels carrying a chosen electronic propagator: Berry parallel transport on adiabatic channels, a multi-state propagator on non-adiabatic ones \cite[Def.~9.16]{Han2026}.

\begin{definition}[Electronic structure level {\cite[Def.~9.17]{Han2026}}]
\label{def:L6}
An object of $\Lk_6(P)$ is an $\Lk_5(P)$-object equipped with an electronic lift
$(\Hel,A,\Xseam,\eta_B)$
satisfying the electronic-structure-origin and $\Lk_5$-consistency conditions;
the forgetful functor $U_6\colon\Lk_6(P)\to\Lk_5(P)$ discards the lift,
retaining only the scalar Morse triple.
\end{definition}

This extension differs from every earlier one in a way worth stating plainly.
At the lower levels a single functor out of the fixed base produced the level freely; here there is no canonical functor $\Lk_5(P)\to\HilbBundR$.
The active rank, the adiabatic region, the connection, the seam, and the propagator on each channel are electronic-structure data the scalar shadow does not determine \cite[\S9.4.3]{Han2026}.
The lift functor $F_6\colon\Lk_6(P)\to\HilbBundR$ thus records a \emph{chosen} lift and does not factor through $U_6$: it cannot be reconstructed from $V$.
In practice this is the statement that locating a seam and evaluating $\eta_B$ is a multi-state computation, needing at least the first two adiabatic surfaces and a method sensitive to their degeneracy such as CASSCF or MRCI \cite{domcke2012role}, whereas any single-reference ground-state method returns only the $\Lk_5$ datum \cite[\S9.4]{Han2026}.

\paragraph{The forcing.}
As in Section~\ref{sec:forcing}, the forgetful functor induces by restriction $\varphi_6\colon\Aut(\Lk_6(P))\to\Aut(\Lk_5(P))$, and the exact sequence
\[
1\to\ker\varphi_6\to\Aut(\Lk_6(P))\xrightarrow{\ \varphi_6\ }\Aut(\Lk_5(P))
\to\coker(\varphi_6)\to 1
\]
has non-trivial cokernel \cite[Eq.~(24),\S8.7]{Han2026}.
As at $\Lk_5$, this is an object refinement, not the separation of a conflated pair.
A single $\Lk_5$-object, the Morse triple $(\Ce(G),V,g)$, admits electronic lifts of distinct bundle topology, indexed by $\eta_B\in H^1(\Omega\setminus\Xseam;\ZZ/2)$, and the standing $\Lk_5$ assumption retains only the CI-free region on which $\eta_B$ vanishes.
Since $\eta_B$ is not a datum of $\Lk_5$, it cannot be read off $(\Ce(G),V,g)$, so the one $\Lk_5$-object splits into topologically distinct $\Lk_6$-objects; the obstruction to a common preimage is the non-trivial element of $\coker(\varphi_6)$ \cite[\S8.7]{Han2026}.
The shape of the argument is that of Section~\ref{sec:forcing}; what has changed across the whole tower is the datum that now fails to be preserved.
At $\Lk_0$ it was a label on a reaction.
Here it is a cohomology class of an electronic bundle: $\Lk_5$ saw the ground state as a function, and $\Lk_6$ sees it as a section of a bundle with non-trivial holonomy, the deepest step the tower takes \cite[\S9.4]{Han2026}.

\subsection{Outlook: full quantum structure}
\label{sec:roadmap}

The electronic-structure level is not the top of the tower.
There is an ordinary chemical distinction it cannot represent.
Take the simplest molecule, $\ce{H2}$.
In its electronic and vibrational ground state.
Ordinary hydrogen is in fact a mixture of two species,
\emph{ortho}- and \emph{para}-hydrogen,
that differ in their rotational spectra and low-temperature heat capacities and that can be separated and stored. 
Their interconversion is so slow in the isolated molecule that para-enriched hydrogen is a stockroom reagent.
That metastability is the resource behind parahydrogen-induced polarisation
\cite{BowersWeitekamp1986, BowersWeitekamp1987}
and signal amplification by reversible exchange \cite{Adams2009},
in which the stored spin order enhances magnetic-resonance signals by orders of magnitude.

The two species, and the rule keeping them apart, come from a symmetry constraint.
The total wavefunction factorises into electronic, vibrational, rotational,
and nuclear-spin parts, and the two protons are identical fermions,
so the product must be antisymmetric under their exchange.
In the ground state the electronic and vibrational factors are symmetric,
so the constraint falls on the remaining pair.
The rotational level of parity $(-1)^J$ must combine with a nuclear-spin state of opposite parity.
Therefore, the Even-$J$ rotation pairs only with the antisymmetric spin singlet (para-hydrogen)
and the odd-$J$ rotation only with the symmetric triplet (ortho-hydrogen) \cite{BonhoefferHarteck1929}.
Exchanging the two species would require changing the nuclear-spin state,
which no spatial interaction does, and that is why they persist as distinct substances.
This is a selection rule forced by an exchange symmetry,
the same shape of argument the tower has run at every level below.

Naming $\Lk_7$ is what lets the rest of the tower be read honestly.
It makes precise which approximation each lower level rests on,
since $\Lk_5$ and $\Lk_6$ assume the Born--Oppenheimer separation and $\Lk_7$ is the level at
which that separation must be derived rather than assumed, and it locates,
by the same forcing question used throughout,
the distinctions lying strictly above electronic structure:
isotopic identity, nuclear-spin statistics, and the individuation of a molecule under identical-particle exchange. 
That last step is not completed here,
and the obstacle is not specific to this work:
recovering a definite molecular structure from an all-particle quantum description is a long-standing open problem,
since the exact bound states of the Coulomb Hamiltonian inherit its full symmetry and so carry no fixed nuclear shape \cite{Woolley1978},
and while structure has been proposed to emerge as a superselection sector in the mass-ratio limit $\varepsilon \to 0$ \cite{Amann1991},
no general theorem identifies an arbitrary chemical graph with such a sector.
Accordingly, $\Lk_7$ is set out in \cite[\S10]{Han2026} as conjectures rather than theorems;
what the present paper claims is the diagnostic and its uniform development across the lower tower,
proved for the base in Section~\ref{sec:forcing} and reused through mechanism,
together with the geometric and electronic lifts whose detailed constructions it imports from \cite{Han2026}.
\section{Application: placing physics in the tower}
\label{sec:application}

The theory section constructed an apparatus; this section uses it. The use is of
two kinds. First, the tower is a \emph{placement procedure}: given a physical
quantity attached to a reaction network --- a heat, a rate, a mechanism, a
barrier, a Berry phase --- it says at which level the quantity lives, what
categorical object it becomes there, and what the tower then guarantees or forbids
about it (Section~\ref{sec:recipe}). Second, the procedure is run end to end on a
single reaction that carries genuine content at every level: the electrocyclic
interconversion of cyclobutene and 1,3-butadiene, the textbook
Woodward--Hoffmann system (Section~\ref{sec:cyclobutene}). At the levels the paper
constructs in full it does real work --- the Deficiency Zero Theorem becomes a
statement about a single forgetful fibre (Section~\ref{sec:dzt}), and the
concerted/stepwise distinction becomes a statement about derivation length --- and
at $\Lk_5$ and $\Lk_6$ it runs the same machinery on geometric and electronic
data, with the constructions imported from \cite{Han2026} as in the theory
section.

\subsection{The placement procedure}
\label{sec:recipe}

The procedure is the following. Confronted with a physical datum $Q$ attached to
the reactions of a network $P$, ask in turn what kind of quantity $Q$ is; the
answer fixes a level, a functor, and a categorical object, and the tower then
supplies a guarantee. The steps are uniform: each asks whether $Q$ is a
\emph{new} kind of distinction relative to the data already placed, and if so
records it as the decoration or lift that the corresponding level adjoins.

\begin{enumerate}[label=(\arabic*),leftmargin=2.4em]
\item \textbf{Is $Q$ a count, a conservation law, or a deficiency?} Then $Q$ is
$\Lk_0$ data: an invariant of the stoichiometric category $\Lk_0(P)$, computable
from the presentation alone (Section~\ref{sec:L0}). The tower guarantees that no
choice of energetic, equilibrium, or kinetic data can change it, since all of
those are decorations \emph{over} $\Lk_0(P)$. Conservation laws are object
potentials; deficiency is the integer $\delta=n-\ell-\rho$.

\item \textbf{Is $Q$ a heat, additive over composition?} Then $Q$ is the
$\Lk_1$ functor $\FH\colon\Lk_0(P)\to\BR$, a labelling of reactions by reals with
composites sent to sums (Section~\ref{sec:dict}). The tower distinguishes two
strengths: plain additivity (Hess bookkeeping) holds for every such $Q$ by
functoriality, while \emph{path-independence} holds only when $Q$ is a coboundary
--- induced by an object potential, i.e.\ a state function. Test which, and the
tower tells you whether your $Q$ is a genuine state function or merely additive.

\item \textbf{Is $Q$ an entropy or a free energy, with a reversal involution?}
Then $Q$ is $\Lk_2$ data: the functor $\FS$, the temperature family $\FG^T$, and
--- when $P$ is reversible --- the $\dagger$-structure recording reversal. Reversal
symmetry ($\FG^T(r^\dagger)=-\FG^T(r)$) and the Wegscheider cycle conditions are
\emph{two} guarantees, not one: the first from $\dagger$-compatibility, the second
again from the coboundary condition. The tower keeps them separate.

\item \textbf{Is $Q$ a rate?} Then $Q$ is the $\Lk_3$ kinetic decoration
$\FP\colon\Lk_0(P)\to\Bg$, a per-reaction channel generator with additive
bookkeeping (Section~\ref{sec:dict}) --- this placement holds for any rate,
whatever the network. What the tower can then add is conditional on the
\emph{base}, not on $Q$: \emph{if} the underlying $\Lk_0(P)$ happens to have
$\delta=0$ and be weakly reversible, the deficiency-zero diagnostic of
Section~\ref{sec:dzt} applies, and \emph{every} value of the rate datum gives the
same qualitative steady-state structure. The placement of $Q$ at $\Lk_3$ and the
deficiency guarantee on the base are two separate facts; conflating them is itself
one of the level-mismatch errors Section~\ref{sec:dzt} catches.

\item \textbf{Does $Q$ distinguish two reactions with the same rate?} Then $Q$ is
not a decoration of $\Lk_0(P)$ at all: it is $\Lk_4$ structural data, recorded by
the double-pushout derivation that realises the reaction (Section~\ref{sec:L4}).
The tower's guarantee is that mechanism is not a function of the lower profile:
the same effective channel may be realised by derivations of different length, and
the worked instance of Section~\ref{sec:cyclobutene} exhibits a pair that
agrees on every datum through $\Lk_3$ and differs only in derivation structure.

\item \textbf{Is $Q$ geometric --- a barrier, an isotope effect, a transition-state
shape?} Then $Q$ is $\Lk_5$ data: the geometric functor $\FV$ into $\OrbMorse$,
carrying a Morse triple $(\Ce(G),V,g)$ (Section~\ref{sec:L5}). The placement is
the same operation as below $\Lk_4$ --- locate $Q$ as the datum a forgetful map
discards --- but the construction of $\FV$ is imported from \cite{Han2026}. The
secondary kinetic isotope effect of Section~\ref{sec:cyclobutene} runs it.

\item \textbf{Is $Q$ topological --- a photochemical branching, a sign holonomy, a
half-integer quantum number?} Then $Q$ is $\Lk_6$ data: the electronic lift $F_6$
into $\HilbBundR$, carrying a $\ZZ/2$ Berry class
$\eta_B=w_1(\Lzero^{\RR})$ over the configuration orbifold (Section~\ref{sec:L6}),
again with the construction imported from \cite{Han2026}. The thermal-versus-
photochemical inversion of Section~\ref{sec:cyclobutene} runs it.
\end{enumerate}

There is a seam in the procedure, and it is worth stating plainly rather than
hiding. Steps (1)--(5) place a quantity using constructions carried out in full in
the theory section: the reader can execute them with this paper alone. Steps
(6)--(7) place a quantity using the same categorical operation --- it is the datum
that the forgetful map to the level below cannot see --- but the target categories
$\OrbMorse$ and $\HilbBundR$, and the proof that the relevant functors are
well-defined, are imported from \cite{Han2026}. The procedure is uniform; what
changes at the seam is not the operation but where its justification is to be
found. The two worked instances at $\Lk_5,\Lk_6$ below are included precisely to
show that the operation does run end to end on real chemistry.

\subsection{Deficiency Zero as a statement about one fibre}
\label{sec:dzt}

The Deficiency Zero Theorem (DZT) is the central result of classical chemical
reaction network theory \cite{HornJackson,Feinberg}. Read through the tower it
acquires a precise structural shape, and that shape resolves two standard misuses
at a stroke. The reading is due to \cite[\S5.6]{Han2026}; it is recalled here
because it is the cleanest demonstration of the placement procedure doing work.

The theorem's hypotheses and its conclusion live at different levels. Weak
reversibility and the deficiency $\delta=0$ are properties of the Petri net $P$
alone --- the complex count $n$, the linkage-class count $\ell$, the stoichiometric
rank $\rho$, and the weak-reversibility condition on the reaction graph are all
invariants of $\Lk_0(P)$, computable without reference to any rate constant. The
conclusion --- existence, uniqueness, and local asymptotic stability of a positive
steady state in each stoichiometric class --- has no meaning at $\Lk_0$ at all:
there is no dynamics without the kinetic functor $\FP$. The new content is therefore
entirely on the conclusion side, at $\Lk_3$.

\begin{proposition}[DZT as fibre-rigidity {\cite[\S5.6]{Han2026}}]
\label{prop:dzt}
Let $U_3\colon\Lk_3(P)\to\Lk_0(P)$ be the forgetful map of
Section~\ref{sec:dict}, stripping the kinetic decoration. Over a fixed base ---
a fixed stoichiometry $\Lk_0(P)$ --- its fibre is the space of all positive
mass-action rate-constant assignments,
\[
U_3^{-1}(\Lk_0(P))\;\cong\;\mathrm{Map}(\Rx,\RR_{>0}),
\]
a copy of $\RR_{>0}$ for each reaction. If $\Lk_0(P)$ satisfies $\delta=0$ and weak
reversibility, then the steady-state conclusion of the DZT holds at \emph{every}
point of this fibre: the assignment $k\mapsto c^{*}(k)$ sending a rate datum to its
unique positive steady state per stoichiometric class is a well-defined map on the
entire fibre. The DZT is thus a statement that the $U_3$-fibre over a
deficiency-zero weakly-reversible base is \emph{rigid}: its qualitative
steady-state structure is constant across the fibre.
\end{proposition}

The phrase ``for every choice of rate constants'' in the classical statement is
exactly this fibre-rigidity: the conclusion is invariant under any reparametrisation
within $U_3^{-1}(\Lk_0(P))$. Equivalently, in the automorphism sequence at $\Lk_3$,
no element of $\coker(\varphi_3)$ --- the rate-constant swaps that preserve the
$\Lk_2$ data but not the $\Lk_3$ data --- can alter the steady-state count, since
every fibre point already gives exactly one steady state per class
\cite[Rem.~5.25]{Han2026}.

\paragraph{Two misuses become level-mismatch errors.} The structural reading is
not merely tidy; it catches two errors that recur in the applied literature, and
catches them as the \emph{same} kind of mistake --- placing a quantity at the wrong
level.

The first is dropping the weak-reversibility hypothesis --- reading ``deficiency
zero'' as if $\delta=0$ alone delivered the unique stable steady state, and
applying the conclusion to a $\delta=0$ network that is not weakly reversible. The
two hypotheses sit at the same level ($\Lk_0$) but are distinct invariants, and the
rigidity of Proposition~\ref{prop:dzt} is asserted only for bases satisfying
\emph{both}. The placement procedure makes the omission visible at step~(1): the
guarantee is attached to the pair (deficiency, weak reversibility) as $\Lk_0$ data,
and a base meeting only one of them is simply not in the class to which the theorem
applies. Stating $\delta=0$ as though it were the whole hypothesis exports a
guarantee proved for one $\Lk_0$ property to a base lacking the other.

The second is generalising the conclusion past mass action --- applying the DZT's
uniqueness verdict to a network with a non-mass-action rate law. The procedure
refuses this at step (4): the fibre $U_3^{-1}(\Lk_0(P))$ is the space of
\emph{mass-action} rate assignments, and the rigidity of
Proposition~\ref{prop:dzt} is a statement about that fibre only. A non-mass-action
rate law is not a point of this fibre; it is a different kinetic decoration
altogether, over which the theorem says nothing. The error is again a
level-mismatch: a guarantee proved for one fibre, exported to a datum that does
not lie in it.

The stochastic counterpart --- the Anderson--Craciun--Kurtz product-form stationary
distribution for deficiency-zero networks \cite{AndersonCraciunKurtz} --- is the
same diagnosis read over distributions rather than points: the same $\Lk_0$
hypotheses, the same single kinetic hypothesis, the same rigidity across the fibre,
now of the stationary distribution of the chemical master equation rather than of
the deterministic fixed point \cite[\S5.6.2]{Han2026}.

\subsection{One reaction through the whole tower: cyclobutene and butadiene}
\label{sec:cyclobutene}

The placement procedure is best seen run once, end to end, on a single reaction
that carries genuine content at every level. The electrocyclic interconversion of
cyclobutene and 1,3-butadiene is the natural choice: it is the textbook
Woodward--Hoffmann system, recognisable to every chemist, and --- unusually --- it
is non-degenerate at all of $\Lk_0$ through $\Lk_6$. Most systems with clean
electronic-topology content are trivial as networks, and most networks with
interesting structure have no chemically central conical intersection; this one
has both. The discussion below walks it up the tower, applying at each step the
placement of Section~\ref{sec:recipe}. Levels $\Lk_0$--$\Lk_4$ are carried out with
the apparatus of the theory section; levels $\Lk_5$--$\Lk_6$ apply the imported
constructions of \cite{Han2026}, as the procedure's seam requires.

\paragraph{$\Lk_0$: the network, and why its species are species.} Take the
reversible thermal core as the Petri net $P$: three species --- cyclobutene,
cisoid butadiene, and $s$-$trans$-butadiene --- with two reversible reactions, the
ring opening cyclobutene $\rightleftharpoons$ cisoid-butadiene and the
conformational interconversion cisoid $\rightleftharpoons$ $s$-$trans$. A chemist
will immediately object that the two butadiene conformers are not distinct
\emph{molecules}: they interconvert by torsion about the central C2--C3 bond over a
barrier of only ${\sim}4\,\mathrm{kcal\,mol^{-1}}$ \cite{Anet1979}, far below the
ring-opening barrier. The objection is exactly the point, and the tower's answer is
a timescale separation made precise. At the ${\sim}420\,\mathrm{K}$ of thermal ring
opening, conformational interconversion runs at ${\sim}10^{9}\text{--}10^{10}\,
\mathrm{s^{-1}}$ while ring opening runs at ${\sim}10^{-3}\,\mathrm{s^{-1}}$ ---
twelve to thirteen orders of magnitude slower. The conformers are therefore a fast
sub-system in permanent quasi-equilibrium relative to the slow reaction that
consumes them, and treating them as two states joined by a fast reversible edge is
the standard pre-equilibrium modelling of chemical kinetics, not a sleight of hand.
(For completeness: the genuine cisoid minimum is the non-planar $s$-$gauche$
conformer near a $34^\circ$ dihedral, with planar $s$-$cis$ a low saddle between
its two enantiomers \cite{Baraban2018}; ``cisoid'' below means this minimum, and
the ring-opening motion delivers butadiene first into this cisoid region because
the four-membered ring geometry maps onto it.)

Granting the two conformers as species, the network is combinatorially explicit.
Each species is its own complex, so the complex count is $n=3$; the reaction graph
is connected, so the linkage-class count is $\ell=1$; the two reaction vectors are
independent, so the stoichiometric rank is $\rho=2$. The deficiency is therefore
\[
\delta=n-\ell-\rho=3-1-2=0 .
\]
This is $\Lk_0$ data in the sense of step~(1): an invariant of $\Lk_0(P)$, fixed
before any rate, enthalpy, or geometry is named, and itself dependent on the
species-counting just defended --- a different coarse-graining (lumping the
conformers) gives a different presentation, which is why the choice had to be made
honestly. The network is weakly reversible (every reaction is reversible and the
graph is strongly connected), so it is a deficiency-zero weakly-reversible network,
and the diagnostic of Section~\ref{sec:dzt} will apply to it directly. The
photochemical branch to bicyclobutane is deliberately excluded from $P$: it is an
effectively irreversible, excited-state-only edge whose inclusion would break weak
reversibility, and it belongs to the $\Lk_6$ regime below rather than to the
ground-state network.

\paragraph{$\Lk_1$--$\Lk_2$: energetics and equilibrium.} The ring opening is
exothermic. From the gas-phase heats of formation --- cyclobutene
$+37.5\,\mathrm{kcal\,mol^{-1}}$ \cite{WibergFenoglio}, 1,3-butadiene
$+26.3\,\mathrm{kcal\,mol^{-1}}$ \cite{ProsenRossini} --- the reaction enthalpy is
\[
\FH \;=\; 26.3-37.5 \;\approx\; -11\ \mathrm{kcal\,mol^{-1}},
\]
a value independently corroborated computationally \cite{NavaCarissan}. (The
figure $28.5\,\mathrm{kcal\,mol^{-1}}$ quoted in some secondary literature is not
the calorimetric enthalpy difference and is not used here.) The chemistry behind
this number is worth stating, because it is not the obvious one. Cyclobutene
carries about $30\,\mathrm{kcal\,mol^{-1}}$ of ring strain \cite{BachDmitrenko},
and ring opening relieves all of it; one might expect a reaction exothermic by that
much. It is not, because opening also \emph{destroys} bonding: the strong
four-membered-ring $\sigma$ framework and the localised C=C $\pi$ bond are traded
for the delocalised but individually weaker $\pi$ system of the conjugated diene.
The ${\sim}30\,\mathrm{kcal\,mol^{-1}}$ of strain relief is largely cancelled by
this loss of $\sigma$/$\pi$ bond strength, leaving the modest
${\sim}11\,\mathrm{kcal\,mol^{-1}}$ net. Placed by step~(2), $\FH$ is the $\Lk_1$
functor value on the ring-opening generator; the residual exothermicity, together
with the entropy gain on opening a ring, makes opening favourable and the reverse
thermal closure unobserved. At $\Lk_2$, the reversal involution $\dagger$ exchanging
each reaction with its reverse is present by construction (the core is reversible),
and the cisoid/$s$-$trans$ equilibrium --- $s$-$trans$ favoured by about
$2.9\,\mathrm{kcal\,mol^{-1}}$ \cite{MuiGrunwald} across the low rotational barrier
above --- is the equilibrium datum on the conformational generator. Reversal
symmetry and the cycle conditions are the two separate $\Lk_2$ guarantees of
step~(3); here the network carries no non-trivial cycle, so only the
$\dagger$-content is in play.

\paragraph{$\Lk_3$: kinetics and the deficiency diagnostic.} The thermal ring
opening of the parent is a clean unimolecular reaction with Arrhenius parameters
$E_a\approx 32.5\,\mathrm{kcal\,mol^{-1}}$, $\log A\approx 13.1$
\cite{CooperWalters}. This is the $\Lk_3$ decoration $\FP$ of step~(4): a rate
constant on each generator, a point of the fibre $U_3^{-1}(\Lk_0(P))$. Because the
underlying $\Lk_0(P)$ has $\delta=0$ and is weakly reversible,
Proposition~\ref{prop:dzt} applies without further work: for \emph{every} choice of
rate constants in this fibre --- not merely the measured ones --- the mass-action
dynamics relaxes to a unique positive equilibrium in each stoichiometric class.
The qualitative kinetic conclusion is thus fixed by the $\Lk_0$ data alone; the
measured $E_a$ locates the system within the rigid fibre but cannot change the
verdict. This is the deficiency diagnostic of Section~\ref{sec:dzt} discharged on a
concrete network.

\paragraph{$\Lk_4$: mechanism, as derivation length.} The thermal ring opening is
concerted, conrotatory, and stereospecific above $99.99\%$ \cite{BraumanArchie}.
The orbital reason is the textbook one and worth stating, because it is what makes
the mechanism a structural rather than a numerical fact. As the C3--C4 $\sigma$
bond breaks, its two lobes must rotate to become the terminal $p$-orbitals of the
diene $\pi$ system; rotating them in the same sense (conrotatory) threads the
four-electron transition state into a M\"obius array --- one phase inversion around
the cycle --- which for $4n$ electrons is aromatic and hence low in energy, whereas
the opposite (disrotatory) sense gives a H\"uckel array that for four electrons is
antiaromatic and high. Equivalently, under the $C_2$ axis preserved by conrotation
every filled cyclobutene orbital correlates with a filled butadiene orbital
($\sigma\to\psi_2$, $\pi\to\psi_1$), so no occupied level is forced uphill, while
the mirror plane preserved by disrotation forces an occupied-to-virtual crossing
and a symmetry-imposed barrier. This is why the stereochemistry is essentially
perfect and electronic in origin.

Step~(5) places the concerted/stepwise question at $\Lk_4$. The concerted route is a
\emph{single} chemical DPO rule $p_{\mathrm{conc}}$ realising the ring opening in
one elementary derivation. A genuinely stepwise route --- bond homolysis to a
discrete diradical, then a second step --- would be a \emph{composite}
$p_2\circ p_1$ of two rules through an intermediate species. By the freeness of
$\Lk_4(P)$ (Proposition~\ref{prop:L4free}) the rule-multiset of a morphism is an
invariant of the symmetric monoidal theory, so $\{p_{\mathrm{conc}}\}$ and
$\{p_1,p_2\}$ are distinct morphisms with the same source and target whenever both
realisations occur: exactly the refinement of Proposition~\ref{prop:Dmech}. The two
share their entire profile through $\Lk_3$ --- same complexes, same effective rate
law after coarse-graining --- and are separated only by derivation length, which is
$\Lk_4$ data and no decoration of $\Lk_0(P)$.

Two cautions keep this honest, and both are instructive for the placement. First,
\emph{concerted-versus-stepwise} (how many elementary rules, an $\Lk_4$ question)
and \emph{conrotatory-versus-disrotatory} (the relative sense of the two terminal
rotations within a single concerted rule, a stereochemical refinement at
$\Lk_{4.5}$) are independent axes; the parent thermal reaction is concerted
\emph{and} conrotatory, and the tower assigns the two distinctions to different
levels rather than conflating them. Second, the parent reaction is not in fact
stepwise: recent analysis finds that the formally ``forbidden'' disrotatory saddle
is usually a second-order saddle point rather than a true transition state, and the
relevant diradical is a transition state for internal rotation of the
\emph{already-opened} product, not a ring-opening intermediate \cite{Carpenter2025}.
So the concerted/stepwise forcing pair is realised cleanly only in systems where a
true stepwise channel exists; for the parent, $\Lk_4$ records that the single
concerted rule is the only realisation, which is itself the $\Lk_4$ datum. The
torquoselectivity of $3$-substituted cyclobutenes --- $\pi$-donors rotate outward to
avoid a filled--filled clash with the high-lying breaking-$\sigma$ orbital,
$\pi$-acceptors rotate inward to accept density from it, with energetic preferences
reaching $19\,\mathrm{kcal\,mol^{-1}}$ \cite{DolbierKoroniakHoukSheu} --- is the
$\Lk_{4.5}$ refinement of \emph{which} conrotatory rule is realised, and is
electronic, not steric, in origin.

\paragraph{$\Lk_5$: the secondary kinetic isotope effect, geometrically.} Isotopic
substitution is invisible to the bond graph: the molecular-graph label records
element, not mass, so cyclobutene and its $3,4$-dideuterio isotopologues are the
same object through $\Lk_{4.5}$. Yet they ring-open at measurably different rates,
and --- the sharper point --- the two diastereotopic C--H positions at each
methylene, the one rotating \emph{inward} and the one rotating \emph{outward} in
the conrotatory motion, carry \emph{distinct} secondary deuterium effects:
$k_{\mathrm H}/k_{\mathrm D}\approx 1.22$ and $1.21$ at
$139.5\,^\circ\mathrm{C}$ for the $cis$- and $trans$-dideuterio isomers (about
$1.10$ per D), with $3$-monodeuteriocyclobutene giving $Z$- over
$E$-deuteriobutadiene in a ratio near $1.10$ \cite{BaldwinKIE}. The physical origin
is the chemistry that makes this an $\Lk_5$ datum and not a bond-graph one: as the
ring opens, the C3 and C4 carbons rehybridise from $sp^3$ toward $sp^2$, the
out-of-plane C--H bending force constants soften, the zero-point energy difference
between C--H and C--D shrinks at the transition state relative to the reactant, and
a normal secondary effect ($k_{\mathrm H}/k_{\mathrm D}>1$) results. Because the
inward- and outward-rotating positions are diastereotopic --- they rehybridise to
different extents in the $C_2$ transition state, computed to be larger at the outer
position \cite{KallelHouk} --- the two carry different effects. Step~(6) places
this at $\Lk_5$: the distinction is carried by the mass-weighted metric $g$ of the
Morse triple $(\Ce(G),V,g)$, which differs between the isotopologues (equal
potential surface $V$, distinct $g$), evaluated at the conrotatory transition state
--- a $C_2$-symmetric index-1 saddle of $V$, not the higher $C_s$ disrotatory
second-order saddle. Applying the imported functor
$\FV\colon\Lk_{4.5}(P)\to\OrbMorse$ \cite[Def.~8.35]{Han2026}, the isotopologues
are a single object of $\Lk_{4.5}(P)$ --- the graph label being isotope-blind ---
that splits into distinct objects $(\Ce(G),V,g_{\mathrm H})\neq(\Ce(G),V,g_{\mathrm
D})$ of $\OrbMorse$ at $\Lk_5$, equal potential surface $V$ and distinct
mass-weighted metric: an object refinement registering as a non-trivial class of
$\coker(\varphi_5)$ \cite[Prop.~8.36]{Han2026}. As step~(6) warns, the operation is
the familiar one --- the datum a forgetful map discards --- but the well-definedness
of $\FV$ and of $\OrbMorse$ is imported. What $\Lk_3$ could only record as two rate
numbers, $\Lk_5$ constrains as geometry: the effect follows from the mass-weighted
Hessian at the saddle, not from a free parameter.

\paragraph{$\Lk_6$: thermal versus photochemical, as seam encirclement.} The same
bond-graph transformation follows opposite orbital-symmetry rules in two regimes:
thermal ring opening is conrotatory, and the idealised Woodward--Hoffmann rule for
the excited state is disrotatory (real cyclobutene photochemistry is wavelength-
dependent and gives con/dis mixtures \cite{Leigh1996}, so the rule is the clean
limit, not the whole story). What underlies the inversion is electronic-structure
topology that the scalar geometric level cannot see. The two electronic states are
two sheets of a double cone joined along a \emph{conical-intersection seam}
$\Xseam$ --- a branch locus in the configuration orbifold where the ground and
first excited surfaces meet. Away from the seam the ground-state surface $V$ is
perfectly smooth, and $V$ alone cannot record whether a reaction path stays clear
of the seam or wraps around it. That is precisely the distinction between the two
regimes: the thermal reaction proceeds entirely on the lower sheet and \emph{does
not encircle} the seam, whereas the ultrafast photochemical reaction is promoted to
the upper sheet and funnels back through the seam, traversing geometries the
thermal path never visits. The datum that separates them is the parity of seam
encirclement of the reaction loop --- the $\ZZ/2$ Berry sign class
$\eta_B=w_1(\Lzero^{\RR})\in H^1(\Ce(G)\setminus\Xseam;\ZZ/2)$ \cite[\S9.2]{Han2026}
--- recording the sign change a real electronic eigenstate acquires when carried
around the seam, the Longuet--Higgins/Herzberg geometric phase and molecular
Aharonov--Bohm effect \cite{LonguetHiggins,Berry}.

Step~(7) places this at $\Lk_6$ as two electronic lifts distinguished by that
class: the lift with $\eta_B=0$, the seam-free regime that $\Lk_5$ fully captures
(the thermal reaction), and the lift with $\eta_B\neq 0$, whose accessible nuclear
loop encircles the seam (the photochemical reaction). Applying the imported lift
functor $F_6\colon\Lk_6(P)\to\HilbBundR$ \cite[Def.~9.17]{Han2026}, the two share
the same scalar shadow under $\FV$ --- the same smooth $V$ on the complement of the
seam --- so a single $\Lk_5$ object refines into two distinct objects of
$\HilbBundR$ at $\Lk_6$: an object refinement registering as a non-trivial class of
$\coker(\varphi_6)$ \cite[\S8.7]{Han2026}. In the tower's reading the thermal and
photochemical regimes are not one nuclear path carrying two lifts --- they traverse
different geometries --- but two paths against a \emph{shared} branched electronic
structure, separated by how they sit relative to its branch point. The
Woodward--Hoffmann inversion is consistent with, and located by, this seam
topology: the same phase-change argument that places the conical intersection
inside the photochemical loop \cite{LonguetHiggins} reproduces the
thermal-conrotatory/photochemical-disrotatory complementarity. (The phase-change
rule \emph{locates} the intersection and is consistent with the selection rule; the
present paper claims only the placement of the thermal/photochemical distinction as
the $\Lk_6$ datum, not a derivation of the Woodward--Hoffmann rule from the Berry
phase, and no cyclobutene-specific Berry-phase computation is asserted
\cite{Han2026}.)

\paragraph{The tower, on one molecule.} Read upward, the same reaction has been
six different objects: a deficiency-zero network at $\Lk_0$, an exothermic
reversible pair at $\Lk_1$--$\Lk_2$, a rigid kinetic fibre at $\Lk_3$, a pair of
DPO derivations of different length at $\Lk_4$, a pair of Morse triples with
distinct mass-weighted metrics at $\Lk_5$, and a pair of electronic lifts with
distinct seam-encirclement parity at $\Lk_6$. At each step the placement was the
same operation --- locate the datum a forgetful map discards --- and what changed
was only the kind of datum and, above $\Lk_4$, where its construction is to be
found. The electrocyclic reaction every chemist knows is, read through the tower, a
single object stratified into six.
\section{Conclusion}
\label{sec:conclusion}

A reaction network admits a single graded description.
Each level adds one kind of chemical content over the level below, and one question locates it:
\emph{what can the level express that the one beneath it cannot?}
The answer takes three forms.

For the decorated levels, it is a broken symmetry,
the automorphisms of a level that fail to descend,
forming a pointed-set cokernel that measures the new content.
For reaction mechanisms, it is an inequality of derivations,
the underlying category changing rather than acquiring a decoration.
For the geometric and electronic levels it is an object refinement,
a single lower object splitting into distinct upper ones.
The forms differ, but the question is one.
What it locates changes as one climbs,
from a coboundary obstruction in energetics to a $\ZZ/2$ Berry class in electronic structure,
while the operation that locates it does not.

The paper carries this out in full for the levels through the mechanistic level, $\Lk_4$.
The base automorphism group is computed,
the forcing class for each decorated level is identified as a coset in its cokernel,
and mechanism is placed as a base change whose content is the derivation structure no rate datum records.
Above $\Lk_4$, the same question is run against the more elaborate target categories of \cite{Han2026}.
The geometric and electronic levels are stated here as higher lifts,
their forcing read as an object refinement,
with the detailed constructions imported rather than reproved.

The placement procedure of Section~\ref{sec:application} turns the apparatus into a method.
It assigns a physical quantity to its level and reads off what the tower then guarantees.
Its first use recasts the Deficiency Zero Theorem as the rigidity of a single forgetful fibre and exposes two standard misuses as errors of level.
Its second carries one electrocyclic reaction up the tower,
recovering at each step the textbook fact that lives there.

One level is left open above the electronic-structure level.
The full quantum description,
in which the Born–Oppenheimer separation the upper levels assume is derived rather than posited,
would sit above it.
This paper does not construct that level;
the monograph \cite{Han2026} sets it out as conjectures rather than theorems. 
Whether the forcing diagnostic extends to it,
in the object-refinement form the upper levels take,
is the natural next question.

The same ideas are further explored in two ways developed elsewhere \cite{Han2026}.
The tower here deepens the description of a fixed network.
A companion construction lets the network itself vary,
which is what comparing whole families of models requires,
including the architectures used in machine learning.
Because each level is built from explicit functors rather than an informal analogy,
the whole tower can be made computational.
The placement of a quantity then becomes something that a program can carry out.

\bibliographystyle{unsrt}
\bibliography{refs}

@misc{Han2026,
  author       = {Han, Kyunghoon},
  title        = {Categorification of Chemical Reactions: A Bottom-Up Tower
                  from Stoichiometry to Quantum Structure},
  year         = {2026},
  eprint       = {2605.14974},
  archivePrefix= {arXiv},
  primaryClass = {math.CT},
  note         = {arXiv:2605.14974},
  url          = {https://arxiv.org/abs/2605.14974}
}

@article{Akitsu2023,
  title={Category Theory in Chemistry},
  author={Akitsu, Takashiro},
  journal={Compounds},
  volume={3},
  number={2},
  pages={334--335},
  year={2023},
  publisher={MDPI}
}

@article{Lewis1925,
  title={A new principle of equilibrium},
  author={Lewis, Gilbert N},
  journal={Proceedings of the National Academy of Sciences},
  volume={11},
  number={3},
  pages={179--183},
  year={1925}
}

@article{Wegscheider1901,
  title={{\"U}ber simultane Gleichgewichte und die Beziehungen zwischen Thermodynamik und Reaktionskinetik homogener Systeme},
  author={Wegscheider, Rud},
  journal={Zeitschrift f{\"u}r physikalische Chemie},
  volume={39},
  number={1},
  pages={257--303},
  year={1902},
  publisher={De Gruyter Oldenbourg}
}

@article{WoodwardHoffmann1965,
  title={Stereochemistry of electrocyclic reactions},
  author={Woodward, Robert B and Hoffmann, Roald},
  journal={Journal of the American Chemical Society},
  volume={87},
  number={2},
  pages={395--397},
  year={1965},
  publisher={ACS Publications}
}

@article{WoodwardHoffmann1969,
  title={The conservation of orbital symmetry},
  author={Woodward, Robert Burns and Hoffmann, Roald},
  journal={Angewandte Chemie International Edition in English},
  volume={8},
  number={11},
  pages={781--853},
  year={1969},
  publisher={Wiley Online Library}
}

@article{LonguetHiggins1975,
  title={The intersection of potential energy surfaces in polyatomic molecules},
  author={Longuet-Higgins, Hugh Christopher},
  journal={Proceedings of the Royal Society of London. A. Mathematical and Physical Sciences},
  volume={344},
  number={1637},
  pages={147--156},
  year={1975},
  publisher={The Royal Society London}
}

@article{CraciunJinYu2020,
  title={An efficient characterization of complex-balanced, detailed-balanced, and weakly reversible systems},
  author={Craciun, Gheorghe and Jin, Jiaxin and Yu, Polly Y},
  journal={SIAM Journal on Applied Mathematics},
  volume={80},
  number={1},
  pages={183--205},
  year={2020},
  publisher={SIAM}
}

@article{DaggettYangLiuMuechler2024,
  title={Toward a Topological Classification of Nonadiabaticity in Chemical Reactions},
  author={Daggett, Christopher and Yang, Kaijie and Liu, Chao-Xing and Muechler, Lukas},
  journal={Chemistry of Materials},
  volume={36},
  number={8},
  pages={3479--3489},
  year={2024},
  publisher={ACS Publications}
}

@article{Yuan2020,
  title={Observation of the geometric phase effect in the {H} + {HD} $\rightarrow$ {H$_2$} + {D} reaction below the conical intersection},
  author={Yuan, Daofu and Huang, Yin and Chen, Wentao and Zhao, Hailin and Yu, Shengrui and Luo, Chang and Tan, Yuxin and Wang, Siwen and Wang, Xingan and Sun, Zhigang and others},
  journal={Nature communications},
  volume={11},
  number={1},
  pages={3640},
  year={2020},
  publisher={Nature Publishing Group UK London}
}

@article{Hess1840,
  author  = {Hess, Germain Henri},
  title   = {Recherches thermo-chimiques},
  journal = {Bulletin scientifique publi{\'e} par l'Acad{\'e}mie imp{\'e}riale des sciences de Saint-P{\'e}tersbourg},
  volume  = {8},
  pages   = {257--272},
  year    = {1840}
}

@article{Gibbs1875,
  title={On the equilibrium of heterogeneous substances},
  author={Gibbs, Josiah Willard},
  journal={American journal of science},
  volume={3},
  number={96},
  pages={441--458},
  year={1878},
  publisher={American Journal of Science}
}

@phdthesis{Fong2016Thesis,
  archiveprefix = {arXiv},
  author = {Brendan Fong},
  eprint = {1609.05382},
  school = {University of Oxford},
  title = {The Algebra of Open and Interconnected Systems},
  year = {2016}
}

@article{baez2019structured,
  author = {Baez, John C. and Courser, Kenny},
  journal = {Theory and Applications of Categories},
  pages = {1771--1822},
  publisher = {Mount Allison University},
  title = {Structured cospans},
  volume = {35},
  number = {48},
  year = {2020}
}

@article{BaezMaster2020,
  title={Open petri nets},
  author={Baez, John C and Master, Jade},
  journal={Mathematical Structures in Computer Science},
  volume={30},
  number={3},
  pages={314--341},
  year={2020},
  publisher={Cambridge University Press}
}

@article{MeseguerMontanari,
  title={Petri nets are monoids},
  author={Meseguer, Jos{\'e} and Montanari, Ugo},
  journal={Information and computation},
  volume={88},
  number={2},
  pages={105--155},
  year={1990},
  publisher={Elsevier}
}

@article{BaezPollard,
  title={A compositional framework for reaction networks},
  author={Baez, John C and Pollard, Blake S},
  journal={Reviews in Mathematical Physics},
  volume={29},
  number={09},
  pages={1750028},
  year={2017},
  publisher={World Scientific}
}

@article{LackSobocinski,
  title={Adhesive and quasiadhesive categories},
  author={Lack, Stephen and Soboci{\'n}ski, Pawe{\l}},
  journal={RAIRO-Theoretical Informatics and Applications},
  volume={39},
  number={3},
  pages={511--545},
  year={2005},
  publisher={EDP Sciences}
}

@article{HornJackson,
  title={General mass action kinetics},
  author={Horn, Fritz and Jackson, Roy},
  journal={Archive for rational mechanics and analysis},
  volume={47},
  number={2},
  pages={81--116},
  year={1972},
  publisher={Springer}
}

@article{Feinberg,
  title={Chemical reaction network structure and the stability of complex isothermal reactors—I. The deficiency zero and deficiency one theorems},
  author={Feinberg, Martin},
  journal={Chemical engineering science},
  volume={42},
  number={10},
  pages={2229--2268},
  year={1987},
  publisher={Elsevier}
}

@article{AndersonCraciunKurtz,
  title={Product-form stationary distributions for deficiency zero chemical reaction networks},
  author={Anderson, David F and Craciun, Gheorghe and Kurtz, Thomas G},
  journal={Bulletin of mathematical biology},
  volume={72},
  number={8},
  pages={1947--1970},
  year={2010},
  publisher={Springer}
}

@article{Berry,
  title={Quantal phase factors accompanying adiabatic changes},
  author={Berry, Michael Victor},
  journal={Proceedings of the Royal Society of London. A. Mathematical and Physical Sciences},
  volume={392},
  number={1802},
  pages={45--57},
  year={1984},
  publisher={The Royal Society London}
}

@article{LonguetHiggins,
  title={The intersection of potential energy surfaces in polyatomic molecules},
  author={Longuet-Higgins, Hugh Christopher},
  journal={Proceedings of the Royal Society of London. A. Mathematical and Physical Sciences},
  volume={344},
  number={1637},
  pages={147--156},
  year={1975},
  publisher={The Royal Society London}
}

@article{WibergFenoglio,
  title     = {Heats of formation of {C$_4$H$_6$} hydrocarbons},
  author={Wiberg, Kenneth B and Fenoglio, Richard A},
  journal={Journal of the American Chemical Society},
  volume={90},
  number={13},
  pages={3395--3397},
  year={1968},
  publisher={ACS Publications}
}

@article{ProsenRossini,
  title={Heats of formation and combustion of 1, 3-butadiene and styrene},
  author={Prosen, Edward J and Rossini, Frederick D},
  journal={Journal of Research of the National Bureau of Standards},
  volume={34},
  number={1},
  pages={59},
  year={1945},
  publisher={National Institute of Standards and Technology (NIST)}
}

@article{NavaCarissan,
  title={On the ring-opening of substituted cyclobutene to benzocyclobutene: analysis of $\pi$ delocalization, hyperconjugation, and ring strain},
  author={Nava, Paola and Carissan, Yannick},
  journal={Physical Chemistry Chemical Physics},
  volume={16},
  number={30},
  pages={16196--16203},
  year={2014},
  publisher={Royal Society of Chemistry}
}

@article{BachDmitrenko,
  title={Strain energy of small ring hydrocarbons. Influence of C- H bond dissociation energies},
  author={Bach, Robert D and Dmitrenko, Olga},
  journal={Journal of the American Chemical Society},
  volume={126},
  number={13},
  pages={4444--4452},
  year={2004},
  publisher={ACS Publications}
}

@article{MuiGrunwald,
  title={Thermodynamics of conformational change in 1, 3-butadiene studied by high-temperature ultraviolet absorption spectroscopy},
  author={Mui, Philip W and Grunwald, Ernest},
  journal={Journal of the American Chemical Society},
  volume={104},
  number={24},
  pages={6562--6566},
  year={1982},
  publisher={ACS Publications}
}

@article{Anet1979,
  title={Conformational analysis of 1, 3-butadiene},
  author={Lipnick, Robert L and Garbisch, Edgar W},
  journal={Journal of the American Chemical Society},
  volume={95},
  number={19},
  pages={6370--6375},
  year={1973},
  publisher={ACS Publications}
}

@article{Baraban2018,
author = {Baraban, Joshua H. and Martin-Drumel, Marie-Aline and Changala, P. Bryan and Eibenberger, Sandra and Nava, Matthew and Patterson, David and Stanton, John F. and Ellison, G. Barney and McCarthy, Michael C.},
title = {The Molecular Structure of gauche-1,3-Butadiene: Experimental Establishment of Non-planarity},
journal = {Angewandte Chemie International Edition},
volume = {57},
number = {7},
pages = {1821-1825},
year = {2018}
}

@article{CooperWalters,
  title={The thermal isomerization of cyclobutene1, 2},
  author={Cooper, Walter and Walters, WD},
  journal={Journal of the American Chemical Society},
  volume={80},
  number={16},
  pages={4220--4224},
  year={1958},
  publisher={ACS Publications}
}

@article{Kurtz1970,
  title={Solutions of ordinary differential equations as limits of pure jump Markov processes},
  author={Kurtz, Thomas G},
  journal={Journal of applied Probability},
  volume={7},
  number={1},
  pages={49--58},
  year={1970},
  publisher={Cambridge University Press}
}

@article{Kurtz1972,
  title={The relationship between stochastic and deterministic models for chemical reactions},
  author={Kurtz, Thomas G},
  journal={The Journal of Chemical Physics},
  volume={57},
  number={7},
  pages={2976--2978},
  year={1972},
  publisher={American Institute of Physics}
}

@article{BraumanArchie,
  title={Energies of alternate electrocyclic pathways. Pyrolysis of cis-3, 4-dimethylcyclobutene},
  author={Brauman, John I and Archie Jr, William C},
  journal={Journal of the American Chemical Society},
  volume={94},
  number={12},
  pages={4262--4265},
  year={1972},
  publisher={ACS Publications}
}

@article{DolbierKoroniakHoukSheu,
  title={Electronic control of stereoselectivities of electrocyclic reactions of cyclobutenes: a triumph of theory in the prediction of organic reactions},
  author={Dolbier, William R and Koroniak, Henryk and Houk, KN and Sheu, Chimin},
  journal={Accounts of chemical research},
  volume={29},
  number={10},
  pages={471--477},
  year={1996},
  publisher={ACS Publications}
}

@inproceedings{Ehrig1973,
  title={Graph-grammars: An algebraic approach},
  author={Ehrig, Hartmut and Pfender, Michael and Schneider, Hans J{\"u}rgen},
  booktitle={14th Annual symposium on switching and automata theory (swat 1973)},
  pages={167--180},
  year={1973},
  organization={IEEE}
}

@article{Carpenter2025,
  title={An argument for abandoning the “allowed” and “forbidden” classification of electrocyclic reactions},
  author={Carpenter, Barry K},
  journal={Chemical Science},
  volume={16},
  number={10},
  pages={4264--4278},
  year={2025},
  publisher={Royal Society of Chemistry}
}

@article{BaldwinKIE,
  title={Diastereotopically distinct secondary deuterium kinetic isotope effects on the thermal isomerization of cyclobutene to butadiene},
  author={Baldwin, John E and Reddy, V Prakash and Schaad, Lawrence J and Hess, B Andes},
  journal={Journal of the American Chemical Society},
  volume={110},
  number={25},
  pages={8555--8556},
  year={1988},
  publisher={ACS Publications}
}

@article{KallelHouk,
  author    = {Storer, Joey W. and Raimondi, Laura and Houk, K. N.},
  title     = {Theoretical secondary kinetic isotope effects and the interpretation of transition state geometries. 2. The Diels-Alder reaction transition state geometry},
  journal   = {Journal of the American Chemical Society},
  volume    = {116},
  number    = {21},
  pages     = {9675--9683},
  year      = {1994},
  publisher = {American Chemical Society}
}

@article{Leigh1996,
author = {Leigh, William J. and Postigo, J. Alberto and Zheng, K.C.},
title = {Cyclobutene photochemistry. Adiabatic photochemical ring opening of alkylcyclobutenes},
journal = {Canadian Journal of Chemistry},
volume = {74},
number = {6},
pages = {951-964},
year = {1996}
}

@book{MacLane1998,
  author    = {Mac Lane, Saunders},
  title     = {Categories for the Working Mathematician},
  series    = {Graduate Texts in Mathematics},
  volume    = {5},
  publisher = {Springer},
  address   = {New York},
  year      = {1978}
}

@book{Leinster2014,
  author    = {Leinster, Tom},
  title     = {Basic Category Theory},
  series    = {Cambridge Studies in Advanced Mathematics},
  volume    = {143},
  publisher = {Cambridge University Press},
  year      = {2014}
}

@article{mayer1996rovibronic,
  title={Rovibronic coupling in the {Na}$_3$ {B} system},
  author={Mayer, M. and Cederbaum, L.S. and K{\"o}ppel, H},
  journal={The Journal of chemical physics},
  volume={104},
  number={22},
  pages={8932--8942},
  year={1996},
  publisher={American Institute of Physics}
}

@article{meiswinkel1991pseudo,
  title={A pseudo-Jahn-Teller treatment of the {B} system of {Na}$_3$},
  author={Meiswinkel, R and K{\"o}ppel, H},
  journal={Zeitschrift f{\"u}r Physik D Atoms, Molecules and Clusters},
  volume={19},
  number={4},
  pages={63--66},
  year={1991},
  publisher={Springer}
}

@article{ahn2019stiefel,
  title={Stiefel--Whitney classes and topological phases in band theory},
  author={Ahn, Junyeong and Park, Sungjoon and Kim, Dongwook and Kim, Youngkuk and Yang, Bohm-Jung},
  journal={Chinese Physics B},
  volume={28},
  number={11},
  pages={117101},
  year={2019}
}

@article{domcke2012role,
  title={Role of conical intersections in molecular spectroscopy and photoinduced chemical dynamics},
  author={Domcke, Wolfgang and Yarkony, David R},
  journal={Annual review of physical chemistry},
  volume={63},
  number={1},
  pages={325--352},
  year={2012},
  publisher={Annual Reviews}
}

@article{BornOppenheimer1927,
  author  = {Born, Max and Oppenheimer, Robert},
  title   = {Zur {Q}uantentheorie der {M}olekeln},
  journal = {Annalen der Physik},
  volume  = {389},
  number  = {20},
  pages   = {457--484},
  year    = {1927}
}

@book{Wolfsberg2009,
  author    = {Wolfsberg, Max and Van Hook, W. Alexander and Paneth, Piotr},
  title     = {Isotope Effects in the Chemical, Geological, and Bio Sciences},
  publisher = {Springer},
  year      = {2009}
}

@book{Kato1966,
  author    = {Kato, Tosio},
  title     = {Perturbation Theory for Linear Operators},
  series    = {Grundlehren der mathematischen Wissenschaften},
  volume    = {132},
  publisher = {Springer},
  year      = {1966}
}

@article{MeadTruhlar1979,
    author = {Mead, C. Alden and Truhlar, Donald G.},
    title = {On the determination of Born–Oppenheimer nuclear motion wave functions including complications due to conical intersections and identical nuclei},
    journal = {The Journal of Chemical Physics},
    volume = {70},
    number = {5},
    pages = {2284-2296},
    year = {1979},
    month = {03}
}

@article{Simon1983,
  title = {Holonomy, the Quantum Adiabatic Theorem, and Berry's Phase},
  author = {Simon, Barry},
  journal = {Phys. Rev. Lett.},
  volume = {51},
  issue = {24},
  pages = {2167--2170},
  numpages = {0},
  year = {1983},
  month = {Dec},
  publisher = {American Physical Society}
}

@article{PanatiSpohnTeufel2003,
  author  = {Panati, Gianluca and Spohn, Herbert and Teufel, Stefan},
  title   = {Space-adiabatic perturbation theory},
  journal = {Advances in Theoretical and Mathematical Physics},
  volume  = {7},
  number  = {1},
  pages   = {145--204},
  year    = {2003}
}

@article{BonhoefferHarteck1929,
  author  = {Bonhoeffer, K. F. and Harteck, P.},
  title   = {{\"U}ber Para- und Orthowasserstoff},
  journal = {Zeitschrift f{\"u}r Physikalische Chemie},
  volume  = {5},
  number  = {1},
  year    = {1929}
}

@article{BowersWeitekamp1986,
  author  = {Bowers, C. Russel and Weitekamp, Daniel P.},
  title   = {Transformation of Symmetrization Order to Nuclear-Spin
             Magnetization by Chemical Reaction and Nuclear Magnetic Resonance},
  journal = {Physical Review Letters},
  volume  = {57},
  number  = {21},
  pages   = {2645},
  year    = {1986}
  }

@article{BowersWeitekamp1987,
  author  = {Bowers, C. Russel and Weitekamp, Daniel P.},
  title   = {Parahydrogen and Synthesis Allow Dramatically Enhanced
             Nuclear Alignment},
  journal = {Journal of the American Chemical Society},
  volume  = {109},
  number  = {18},
  pages   = {5541--5542},
  year    = {1987}
}

@article{Adams2009,
  author  = {Adams, R. W. and Aguilar, J. A. and Atkinson, K. D. and
             Cowley, M. J. and Elliott, P. I. P. and Duckett, S. B. and
             Green, G. G. R. and Khazal, I. G. and L{\'o}pez-Serrano, J. and
             Williamson, D. C.},
  title   = {Reversible Interactions with Para-Hydrogen Enhance NMR Sensitivity
             by Polarization Transfer},
  journal = {Science},
  volume  = {323},
  number  = {5922},
  pages   = {1708--1711},
  year    = {2009}
}

@article{Woolley1978,
  author  = {Woolley, R. G.},
  title   = {Must a Molecule Have a Shape?},
  journal = {Journal of the American Chemical Society},
  volume  = {100},
  number  = {4},
  pages   = {1073--1078},
  year    = {1978}
}

@article{Amann1991,
  author  = {Amann, A.},
  title   = {Chirality: A Superselection Rule Generated by the Molecular
             Environment?},
  journal = {Journal of Mathematical Chemistry},
  volume  = {6},
  number  = {1},
  pages   = {1--15},
  year    = {1991}
}

\end{document}